\documentclass[useAMS,usenatbib,referee]{biom}

\def\bSig\mathbf{\Sigma}

\usepackage[figuresright]{rotating}

\usepackage{color}
\usepackage{amsmath}
\usepackage{amssymb}
\usepackage{booktabs}
\usepackage{nameref}
\usepackage{bigfoot}

\newcommand{\CI}{\mathrel{\perp\!\!\!\perp}}

\newcommand{\Exp}{\mathbb{E}}
\newcommand{\expit}{\mathrm{expit}}
\newcommand{\logit}{\mathrm{logit}}
\newtheorem{appendixcorollary}{Corollary}

\renewcommand*\figureplace{}
\renewcommand*\tableplace{}

\title[Weighting for outcome and exposure misclassification and confounding]{A weighting method for simultaneous adjustment for confounding and joint exposure-outcome misclassifications}

\author{
Bas B.L. Penning de Vries\textsuperscript{1,*}, Maarten van Smeden\textsuperscript{1}, and Rolf H.H. Groenwold\textsuperscript{1,2}
\email{B.B.L.Penning\_de\_Vries@lumc.nl} \\
\textsuperscript{1}Department of Clinical Epidemiology, \textsuperscript{2}Department of Biomedical Data Sciences,\\
Leiden University Medical Center, PO Box 9600, 2300 RC, The Netherlands
}

\begin{document}




\makeatletter
\def\@evenhead{\Large\@ddell\hbox to
0pt{\small\thepage}\hspace{1.5em}\hfill}
 
\gdef\@journal{\vbox{\fontsize{10}{10}\selectfont\vbox{\hsize\textwidth{\sc }
}
}}
\def\ps@titlepage{\let\@mkboth\@gobbletwo
 \def\@oddhead{\footnotesize\@journal\hfill}
 \def\@oddfoot{\hfill{\reset@font\ 
 }\hfill}
 \def\@evenhead{\footnotesize\@journal\hfill}
 \def\sectionmark##1{}
 \def\subsectionmark##1{}}
\makeatother

\label{firstpage}

\begin{abstract}
Joint misclassification of exposure and outcome variables can lead to considerable bias in epidemiological studies of causal exposure-outcome effects. 
In this paper, we present a new maximum likelihood based estimator for the marginal causal odd-ratio that simulaneously adjusts for confounding and several forms of joint misclassification of the exposure and outcome variables.
The proposed method relies on validation data for the construction of weights that account for both sources of bias. The weighting estimator, which is an extension of the exposure misclassification weighting estimator proposed by \citeauthor{Gravel2018} (Statistics in Medicine, \citeyear{Gravel2018}), is applied to reinfarction data. 
Simulation studies were carried out 
to study its finite sample properties and compare it with methods that do not account for confounding or misclassification. The new estimator showed favourable large sample properties in the simulations. Further research is needed to study the sensitivity of the proposed method and that of alternatives to violations of their assumptions. The implementation of the estimator is facilitated by a new R function in an existing R package.\\
\end{abstract}


\maketitle

\section{Introduction}
\label{s:introduction}

In epidemiological research on causal associations between a particular exposure and a certain outcome, erroneous information on either or both of these variables poses a serious methodological obstacle in making valid inferences. In particular, joint misclassification of exposure and outcome can lead to considerable bias of standard causal effect estimators, with direction and magnitude depending on various factors, including the misclassification mechanism and the direction and magnitude of the true effect \citep{Kristensen1992,Brenner1993,Vogel2005,Jurek2008,VanderWeele2012,Brooks2018}.

Exposure and outcome misclassification is typically categorised according to two separate properties: whether or not the misclassification is differential and whether or not it is dependent relative to some covariate vector $L$ containing patient characteristics \citep{Kristensen1992,VanderWeele2012}. Joint misclassification of exposure and outcome is said to be \emph{nondifferential} if (1) the sensitivity and specificity of exposure classification are constant across all categories of the (true) outcome given $L$ and (2) the sensitivity and specificity of outcome classification are constant across all categories of the (true) exposure given $L$; otherwise it is \emph{differential}. Misclassification is said to be \emph{independent} if the joint probability of any exposure and outcome classification given any true exposure and outcome categories and $L$ can be factored into the product of the corresponding probabilities for exposure and outcome separately; otherwise, it is \emph{dependent}. In \citeauthor{Dawid1979}'s notation (\citeyear{Dawid1979}), that is, if true exposure level $A$ and true outcome $Y$ are (potentially mis)classified as $B$ and $Z$, respectively, misclassification is nondifferential if and only if $B\CI Y|A,L$ and $Z\CI A|Y,L$ and independent if and only if $Z\CI B|Y,A,L$. 

Epidemiological research hampered by joint misclassification of some type is likely voluminous \citep{Brooks2018}. Examples of studies affected by exposure and outcome misclassification can be found, for example, in the literature on the causal effects of drug use, which is largely based on routinely collected data, where exposures are typically operationalised on the basis of prescription records and where outcomes are often self-reported \citep{Marcum2013,Culver2012,Leong2013,Ni2017}. Departures from differentiality are likely, particularly when it concerns prescription-only medication.
In applied epidemiological research, misclassification or some of its potential consequences are often ignored 
\citep{Jurek2006,Brakenhoff2018}. The assertion often made in the discussion of study results that observed measures of association are biased toward the null under nondifferentiality, for example, is not generally true unless additional conditions are presupposed \citep{Brenner1993,Brooks2018}.

Methods to adjust for misclassification rely on additional information that can be used to estimate or correct for bias. One potential source of information is validation data obtained through supposedly infallible measurement. Recently, \cite{Gravel2018} proposed an inverse probability weighting (IPW) method to simultaneously address confounding and outcome misclassification by means of internal validation data. In what follows, we propose an extension of this method to allow for confounding adjustment and joint exposure and outcome misclassification. This flexible estimator allows for the misclassifications to be dependent, differential or both. In Section~\ref{s:development}, inverse probability weights for confounding and joint misclassification are introduced through a hypothetical study based on the illustrative example of \citeauthor{Gravel2018} (\citeyear{Gravel2018}). Section~\ref{s:estimation} details methods for estimation of the various components of the proposed weights. In Section~\ref{s:simulation}, we describe a series of Monte Carlo simulations that were used to study properties of the proposed method in finite samples. We conclude with a summary and discussion of our findings in context of the existing literature.

\section{Data distribution for illustration and development of weighting method}\label{s:development}

We first consider the data and setting described by \cite{Gravel2018} and suppose that Table~\ref{t:table1} represents a simple random (i.i.d.) sample from (or that its cell counts are proportional to the respective densities in) the population of interest. This illustration is based on a cohort study on the association between post-myocardial infarction statin use ($A$) and the 1-year risk of reinfarction ($Y$). In what follows, we will refer to this example as the `reinfarction example'.

Throughout we take the counterfactual framework for causal inference, formal accounts of which are given for example by \cite{Neyman1935}, \cite{Rubin1974}, \cite{Holland1986}, \cite{Holland1988} and \cite{Pearl2009}.
The interest, we suppose, lies in estimating a function of counterfactuals, in particular the causal marginal odds ratio OR,
\begin{align}
\text{OR} &= \frac{\Exp[Y(1)]/(1-\Exp[Y(1)])}{\Exp[Y(0)]/(1-\Exp[Y(0)])}, \label{e:or1}
\end{align}
where $Y(0),Y(1)$ denote the counterfactual outcomes for hypothetical interventions setting $A$ to $0$ and $1$, respectively.

\begin{table}
\centering
\caption{Cross-classification of the reinfarction data for 33,007 individuals as given by Gravel and Platt (2018).\break
}\label{t:table1}
\begin{tabular*}{\textwidth}{
@{\extracolsep{\fill}}lrrlrr@{\extracolsep{\fill}}
}
\toprule
&\multicolumn{2}{c}{$L=0$}&&\multicolumn{2}{c}{$L=1$}\\\cline{2-3}\cline{5-6}
&$A=0$&$A=1$&&$A=0$&$A=1$\\\midrule
$Y=0$&$11602$&$13116$&&$1302$&$5363$\\
$Y=1$&$890$&$589$&&$49$&$96$\\
\bottomrule
\end{tabular*}
\end{table}

\subsection{No misclassification}

Under conditional exchangeability given $L$ (i.e., $(Y(0),Y(1))\CI A|L$), consistency ($Y(a)=Y$ if $A=a$) and positivity ($\Pr(A=a|L=l)>0$ for $a=0,1$ and all $l$ in the support of $L$), the odds ratio of \eqref{e:or1} can be expressed in terms of `observables' (meaning, here, variables that would be observed had there been no measurement error) as follows:
\begin{align}
\text{OR} &= \frac{\Exp[WY|A=1]/(1-\Exp[WY|A=1])}{\Exp[WY|A=0]/(1-\Exp[WY|A=0])} \label{e:or2}
\end{align}
with weights $W$ defined as the inverse probability of the allocated exposure level $A$ given $L$ (i.e., the inverse propensity score) multiplied by the prevalence of the allocated exposure level $A$ (i.e., $W=\Pr(A)/\Pr(A|L)$; Appendix~\nameref{s:app_weighting}).

Replacing components of the right-hand side of \eqref{e:or2} with sample analogues, we obtain the following estimator for the setting where $L$ is binary:
\begin{align}
\widehat{\text{OR}} &:= \frac{\widehat{\Exp}[\widehat{W}Y|A=1]/(1-\widehat{\Exp}[\widehat{W}Y|A=1])}{\widehat{\Exp}[\widehat{W}Y|A=0]/(1-\widehat{\Exp}[\widehat{W}Y|A=0])} \nonumber\\
&= \frac{(\widehat{W}_{10}n_{110}+\widehat{W}_{11}n_{111})/(n_{110}+n_{111}+n_{010}+n_{011}-\widehat{W}_{10}n_{110}-\widehat{W}_{11}n_{111})}{(\widehat{W}_{00}n_{100}+\widehat{W}_{01}n_{101})/(n_{100}+n_{101}+n_{000}+n_{001}-\widehat{W}_{00}n_{100}-\widehat{W}_{01}n_{101})}, 
\end{align}
where $n_{yal}$ denotes the number of subjects with $Y=y$, $A=a$, $L=l$ and where $\widehat{W}_{al}$ is the product of the proportion of subjects in the sample with $A=a$ and the inverse of the proportion of subjects with $A=a$ among those with $L=l$. For the data in Table~\ref{t:table1}, we obtain $\widehat{\text{OR}}\approx 0.573$. The corresponding crude odds ratio (i.e., with $\widehat{W}=1$) is $0.509$.

\subsection{Joint misclassification}

Suppose that rather than observing $Y$ and $A$ we observe $Z$ and $B$, the 
misclassified versions of $Y$ and $A$, respectively.
The relation between $Z$ and $B$ on the one hand and $Y$, $A$ and $L$ on the other can be expressed as follows:
\begin{align*}
\Pr(Z=z,B=b|Y=y,A=a,L=l) 
&= (\pi_{byal})^{z}(1-\pi_{byal})^{1-z}(\lambda_{yal})^{b}(1-\lambda_{yal})^{1-b}
\end{align*}
for $z,b\in\{0,1\}$ and all possible realisations $y,a,l$ of $Y,A,L$, and where $\pi_{byal}=\Pr(Z=1|B=b,Y=y,A=a,L=l)$ and $\lambda_{yal}=\Pr(B=1|Y=y,A=a,L=l)$.

To simulate (dependent differential) misclassification in the reinfarction dataset, we use the 
true positive and false positive rates given in Table~\ref{t:table2}. The expected cell counts for these rates are given in Tables~\ref{t:table3}~and~\ref{t:table4}.

\begin{table}
\centering
\caption{True and false positive rates for reinfarction example. For $b,y,a,l\in\{0,1\}$, $\lambda_{yal}=\Pr(B=1|Y=y,A=a,L=l)$ and $\pi_{byal}=\Pr(Z=z|B=b,Y=y,A=a,L=l)$. \break
}\label{t:table2}
\begin{tabular*}{\textwidth}{
@{\extracolsep{\fill}}ccc@{\extracolsep{\fill}}
}
\toprule
$\pi_{0000}=0.050$&$\pi_{0001}=0.020$&$\lambda_{000}=0.010$\\
$\pi_{1000}=0.060$&$\pi_{1001}=0.108$&$\lambda_{100}=0.181$\\
$\pi_{0100}=0.930$&$\pi_{0101}=0.806$&$\lambda_{010}=0.880$\\
$\pi_{1100}=0.938$&$\pi_{1101}=0.692$&$\lambda_{110}=0.910$\\
$\pi_{0010}=0.030$&$\pi_{0011}=0.109$&$\lambda_{001}=0.100$\\ 
$\pi_{1010}=0.060$&$\pi_{1011}=0.050$&$\lambda_{101}=0.265$\\
$\pi_{0110}=0.906$&$\pi_{0111}=0.765$&$\lambda_{011}=0.930$\\
$\pi_{1110}=0.950$&$\pi_{1111}=0.861$&$\lambda_{111}=0.823$\\
\bottomrule
\end{tabular*}
\end{table}

\begin{table}
\centering
\caption{Expected cell counts (rounded to integers) for reinfarction example after misclassification was introduced. Because of rounding, the sum of all cell entries is 33,006 rather than 33,007, the size of the reinfarction dataset.\break
}\label{t:table3}
\begin{tabular*}{\textwidth}{
@{\extracolsep{\fill}}lrrlrr@{\extracolsep{\fill}}
}
\toprule
&\multicolumn{2}{c}{$Z=0$}&&\multicolumn{2}{c}{$Z=1$}\\\cline{2-3}\cline{5-6}
&$B=0$&$B=1$&&$B=0$&$B=1$\\\midrule
$Y=0,~A=0,~L=0$&$10912$&$109$&&$574$&$7$ \\
$Y=1,~A=0,~L=0$&$51$&$10$&&$678$&$151$ \\
$Y=0,~A=1,~L=0$&$1527$&$10850$&&$47$&$693$ \\
$Y=1,~A=1,~L=0$&$5$&$27$&&$48$&$509$ \\
$Y=0,~A=0,~L=1$&$1148$&$116$&&$23$&$14$ \\
$Y=1,~A=0,~L=1$&$7$&$4$&&$29$&$9$ \\
$Y=0,~A=1,~L=1$&$334$&$4738$&&$41$&$249$ \\
$Y=1,~A=1,~L=1$&$4$&$11$&&$13$&$68$ \\
\bottomrule
\end{tabular*}
\end{table}

\begin{table}
\centering
\caption{Expected cell counts (rounded to integers) for illustrative study setting after misclassification, collapsed over $Y$ and $A$.\break
}\label{t:table4}
\begin{tabular*}{\textwidth}{
@{\extracolsep{\fill}}lrrlrr@{\extracolsep{\fill}}
}
\toprule
&\multicolumn{2}{c}{$L=0$}&&\multicolumn{2}{c}{$L=1$}\\\cline{2-3}\cline{5-6}
&$B=0$&$B=1$&&$B=0$&$B=1$\\\midrule
$Z=0$&$12495$&$10996$&&$1347$&$1360$\\
$Z=1$&$1493$&$4869$&&$106$&$340$\\
\bottomrule
\end{tabular*}
\end{table}

We redefine the weights in \eqref{e:or2} as a function of $B$ and $L$ (according to Appendix~\nameref{s:app_weighting}) such that 
\bgroup
\allowdisplaybreaks
\begin{align}
W &= \frac{p(B)\varepsilon_{BL}}
{\sum_{y}\sum_{a} \pi_{ByaL}(\lambda_{yaL})^{B}(1-\lambda_{yaL})^{1-B}(\varepsilon_{aL})^{y}(1-\varepsilon_{aL})^{1-{y}}(\delta_L)^{a}(1-\delta_L)^{1-a}},
\label{e:weights}
\end{align}
where $p(B)$ is the prevalence of level $B$ of the potentially misclassified version of the exposure variable and where $\varepsilon_{al}=\Pr(Y=1|A=a,L=l)$ and $\delta_l=\Pr(A=1|L=l)$ for all possible realisations $a$ and $l$ of $A$ and $L$, respectively. In Appendix~\nameref{s:app_weighting}, it is shown that
\begin{align}
\text{OR} &= \frac{\Exp[WZ|B=1]/(1-\Exp[WZ|B=1])}{\Exp[WZ|B=0]/(1-\Exp[WZ|B=0])}, \label{e:or3}
\end{align}
which suggests the plug-in estimator
\begin{align}
\widehat{\text{OR}} &:= \frac{\widehat{\Exp}[\widehat{W}Z|B=1]/(1-\widehat{\Exp}[\widehat{W}Z|B=1])}{\widehat{\Exp}[\widehat{W}Z|B=0]/(1-\widehat{\Exp}[\widehat{W}Z|B=0])},\label{e:or4}
\end{align}
where $\widehat{\Exp}$ denotes the sample mean operator and $\widehat{W}$ the sample analogue (i.e., consistent estimator) of $W$ in \eqref{e:weights}.

In the absence of exposure misclassification, \eqref{e:weights} reduces to
\begin{align}W=\Bigg(\frac{(\delta_{L})^A(1-\delta_L)^{1-A}}{p(A)}\Bigg[\pi_{A0AL}\frac{1-\varepsilon_{AL}}{\varepsilon_{AL}}+\pi_{A1AL}\Bigg]\Bigg)^{-1}.\label{e:weights_GP}\end{align}
The first term within the round brackets corrects for confounding and represents the propensity score divided by prevalence of exposure level $A$. The term within square brackets is a factor that corrects for misclassification in the outcome variable. This correction factor is similar to that proposed by \cite{Gravel2018}. The only difference is that where in \eqref{e:weights_GP} it does not depend on the fallible measurement $Z$ of $Y$, \citeauthor{Gravel2018} define different weights for subjects with $Z=0$. Note, however, that the choice of weights for subjects with $Z=0$ does not affect the population quantity in \eqref{e:or3} or the estimator defined by \eqref{e:or4}.

As for the reinfarction example, the odds ratio estimate for the exposure-outcome effect based on inverse probability weighting that assumes absence of exposure or outcome misclassification is $1.120$, while the corresponding misclassification naive crude odds ratio is $1.031$. Estimation of the population weights $W$ 
from observables using validation data is discussed in the next section. As shown below, weighting using the proposed weights that account for confounding and outcome and exposure misclassification results in an odds ratio of $\text{OR}=\widehat{\text{OR}}\approx0.573$. Inference based on \eqref{e:weights_GP} rather than \eqref{e:weights}, i.e., ignoring misclassification in the exposure but correcting for outcome misclassification, yields an odds ratio estimate of 0.934.

\subsection{Parameterisation based on positive and negative predictive values}
In the foregoing discussion, the proposed weights were expressed in terms of sensitivity and specificity parameters. The sensitivity and specificity of $Z$ with respect to $Y$, given $(B,A,L)$, are $\pi_{{B}1AL}$ and $1-\pi_{{B}0AL}$, respectively. Similarly, $\lambda_{Y1L}$ and $1-\lambda_{Y0L}$ reflect the sensitivity and specificity, respectively, with respect to $A$, conditional on $Y$ and $L$.

As discussed below, it may be more convenient to choose a parameterisation that is based on (positive and negative) predictive values. Define $\delta^\ast_l=\Pr(B=1|L=l)$, $\varepsilon^\ast_{bl}=\Pr(Z=1|B=b,L=l)$, $\lambda^\ast_{zb l}=\Pr(A=1|Z=z,B=b,L=l)$ and $\pi^\ast_{azbl}=\Pr(Y=1|A=a,Z=z,B=b,L=l)$. The weights in \eqref{e:weights} can be rewritten as
\begin{align}  
W &= \frac{
\sum_{y}\sum_{a} \pi^\ast_{B yaL}(\lambda^\ast_{yaL})^{B}(1-\lambda^\ast_{yaL})^{1-B}(\varepsilon^\ast_{aL})^{y}(1-\varepsilon^\ast_{aL})^{1-{y}}(\delta^\ast_L)^{a}(1-\delta^\ast_L)^{1-a}
}
{
\sum_{y}\sum_{a} (\lambda^\ast_{yaL})^{B}(1-\lambda^\ast_{yaL})^{1-B}(\varepsilon^\ast_{aL})^{y}(1-\varepsilon^\ast_{aL})^{1-{y}}(\delta^\ast_L)^{a}(1-\delta^\ast_L)^{1-a}
}
\nonumber\\&\qquad\times\frac{p(B)}{\varepsilon^\ast_{B L}(\delta^\ast_L)^{B}(1-\delta^\ast_L)^{1-B}}.
\label{e:weights2}
\end{align}
In the absence of exposure misclassification, these weights simplify to {\begin{align}
W=\frac{p(A)}{(\delta_{L})^A(1-\delta_L)^{1-A}}\frac{\varepsilon_{AL}}{\varepsilon^\ast_{AL}} .\nonumber
\end{align}

\section{Estimation of weights}\label{s:estimation}
Estimation of the proposed weights can be done using a number of approaches and we will here consider a maximum likelihood approach that assumes the availability of internal validation data, i.e., that some study participants have their observed exposure or outcome measured by an `infallible' or `gold standard' (100\% accurate) classifier.

\subsection{Validation subset inclusion mechanism}
Let $R_{Y}$ be the indicator variable that takes the value of 1 if the outcome is observed (i.e., measured by an infallible classifier) and 0 otherwise. Similarly, define $R_{A}$ to be the indicator variable that takes the value of 1 if the exposure variable is observed and 0 otherwise. 
$R_{Y}$ and $R_{A}$ reflect which subjects have validation data available on $Y$ and $A$, respectively. The subset of subjects with validation data on $Y$ need not fully overlap with the subset with validation data on $A$.

The validation subsets can be approached from the missing data framework of \cite{Rubin1976}. Provided that $Z,B,L$ are free of missing values, \citeauthor{Rubin1976}'s missing at random (MAR) condition is met whenever the vector $(R_{Y},R_{A})$ is conditionally independent of $(Y,A)$ given $(Z,B,L)$.

\subsection{Full likelihood approach based on parameterisation in terms of sensitivities and specificities}
Simultaneous estimation of the whole vector of $\delta$, $\varepsilon$, $\lambda$ and $\pi$ parameters can be done via maximum likelihood estimation as follows. Assuming i.i.d. observations $(Z_i,B_i,Y_i,A_i,L_i)$ and ignorable missingness in the sense of \cite{Rubin1976} (MAR and distinctness), for valid likelihood-based inference it is appropriate to maximise the following log-likelihood over the parameter space of $\theta$, the vector of $\delta$, $\varepsilon$, $\lambda$ and $\pi$ parameters:
\bgroup\allowdisplaybreaks
\begin{align}
\ell(\theta)
&= \sum_{i:R_{Yi}=R_{Ai}=1}\log f(\theta;Z_i,B_i,Y_i,A_i,L_i) \nonumber\\
&\qquad+ \sum_{i:R_{Yi}=1\wedge R_{Ai}=0}\log \sum_{A_i} f(\theta;Z_i,B_i,Y_i,A_i,L_i)\nonumber\\
&\qquad+ \sum_{i:R_{Yi}=0\wedge R_{Ai}=1}\log \sum_{Y_i} f(\theta;Z_i,B_i,Y_i,A_i,L_i)\nonumber\\
&\qquad+ \sum_{i:R_{Yi}=R_{Ai}=0}\log \sum_{Y_i}\sum_{A_i} f(\theta;Z_i,B_i,Y_i,A_i,L_i),\nonumber
\end{align}\egroup
where
\begin{align}
f(\theta;Z_i,B_i,Y_i,A_i,L_i)&=(\pi_{B_iY_iA_iL_i})^{Z_i}(1-\pi_{B_iY_iA_iL_i})^{1-Z_i}(\lambda_{Y_iA_iL_i})^{B_i}(1-\lambda_{Y_iA_iL_i})^{1-B_i}
\nonumber\\&\qquad\times(\varepsilon_{A_iL_i})^{Y_i}(1-\varepsilon_{A_iL_i})^{1-{Y_i}}(\delta_{L_i})^{A_i}(1-\delta_{L_i})^{1-A_i}.\nonumber
\end{align}
Evaluating this log-likelihood involves marginalising over unobserved quantities in the last three terms of $\ell(\theta)$. The log-likelihood equations may become considerably more tractable if we choose a parameterisation of the likelihood that is based on predictive values rather than sensitivities and specificities.


\subsection{Full likelihood approach based on parameterisation in terms of predictive values}

Inference may alternatively be based on a log-likelihood that is parameterised in terms of the vector $\theta^\ast$ of the $\delta^\ast$, $\varepsilon^\ast$, $\lambda^\ast$ and $\pi^\ast$ parameters, i.e.,
\begin{align*}
\ell^\ast(\theta^\ast) 
&= \sum_{i:R_{Yi}=R_{Ai}=1}\log g(\theta^\ast;Z_i,B_i,Y_i,A_i,L_i) \nonumber\\
&\qquad+ \sum_{i:R_{Yi}=1\wedge R_{Ai}=0}\log \sum_{A_i} g(\theta^\ast;Z_i,B_i,Y_i,A_i,L_i)\nonumber\\
&\qquad+ \sum_{i:R_{Yi}=0\wedge R_{Ai}=1}\log \sum_{Y_i} g(\theta^\ast;Z_i,B_i,Y_i,A_i,L_i)\nonumber\\
&\qquad+ \sum_{i:R_{Yi}=R_{Ai}=0}\log \sum_{Y_i}\sum_{A_i} g(\theta^\ast;Z_i,B_i,Y_i,A_i,L_i),\nonumber
\end{align*}\egroup where
\begin{align}
g(\theta^\ast;Z_i,B_i,Y_i,A_i,L_i)
&=(\pi^\ast_{A_iZ_iB_iL_i})^{Y_i}(1-\pi^{\ast}_{A_iZ_iB_iL_i})^{1-Y_i}(\lambda^\ast_{Z_iB_iL_i})^{A_i}(1-\lambda^\ast_{Z_iB_iL_i})^{1-A_i}
\nonumber\\&\qquad\times(\varepsilon^\ast_{B_iL_i})^{Z_i}(1-\varepsilon^\ast_{B_iL_i})^{1-{Z_i}}(\delta^\ast_{L_i})^{B_i}(1-\delta^\ast_{L_i})^{1-B_i}.\nonumber
\end{align}
If validation data is available on $Y$ if and only if it is available on $A$, the complete data log-likelihood ignoring the missing data mechanism can be conveniently expressed as follows:
\begin{align}
\ell^\ast(\theta^\ast) &= \ell_1^\ast(\theta^\ast)+\ell_2^\ast(\theta^\ast)+\ell_3^\ast(\theta^\ast)+\ell_4^\ast(\theta^\ast), \label{e:loglikpredval}
\end{align}
with $\theta^\ast$ denoting the vector of $\delta^\ast$, $\varepsilon^\ast$, $\lambda^\ast$ and $\pi^\ast$ parameters and where
\begin{align*}
\ell_1^\ast(\theta^\ast)&= \sum_{i:R_{Yi}=R_{Ai}=1} Y_i\log(\pi^\ast_{A_iZ_iB_iL_i})+(1-Y_i)\log(1-\pi^\ast_{A_iZ_iB_iL_i})\\
\ell_2^\ast(\theta^\ast)&= \sum_{i:R_{Yi}=R_{Ai}=1} A_i\log(\lambda^\ast_{Z_iB_iL_i})+(1-A_i)\log(1-\lambda^\ast_{Z_iB_iL_i})\\
\ell_3^\ast(\theta^\ast)&= \sum_{i} Z_i\log(\varepsilon^\ast_{B_iL_i})+(1-Z_i)\log(1-\varepsilon^\ast_{B_iL_i})\\
\ell_4^\ast(\theta^\ast)&= \sum_{i} B_i\log(\delta^\ast_{L_i})+(1-B_i)\log(1-\delta^\ast_{L_i}).
\end{align*}
Now, assuming distinct parameter spaces for the vectors of $\pi^\ast$, $\lambda^\ast$, $\varepsilon^\ast$, and $\delta^\ast$ parameters, the parameter values that maximise $\ell^\ast(\theta^\ast)$ can be found by separately maximising $\ell_1^\ast(\theta^\ast)$ and $\ell_2^\ast(\theta^\ast)$ in the validation subset with respect to the $\pi^\ast$ and $\lambda^\ast$ parameters, respectively, and $\ell_3^\ast(\theta^\ast)$ and $\ell_4^\ast(\theta^\ast)$ in the entire dataset with respect to $\varepsilon^\ast$ and $\delta^\ast$. Following \cite{Gravel2018} and \cite{Tang2013}, the sum of the first and last two terms are therefore suitably labelled the internal validation and main study log-likelihood, respectively. With this parameterisation, finding the maximum likelihood estimates is readily achieved by taking advantage of standard statistical software.

\subsection{Equivalence of likelihood approaches based on different parameterisations}

Without restrictions imposed on \begin{align*}\theta_l&:=(\pi_{000l},\pi_{100l},\pi_{010l},\pi_{110l},\pi_{001l},\pi_{101l},\pi_{011l},\pi_{111l},\lambda_{00l},\lambda_{10l},\lambda_{01l},\lambda_{11l},\varepsilon_{0l},\varepsilon_{1l},\delta_{l})\text{~~and}\\\theta^\ast_l&:=(\pi^\ast_{000l},\pi^\ast_{100l},\pi^\ast_{010l},\pi^\ast_{110l},\pi^\ast_{001l},\pi^\ast_{101l},\pi^\ast_{011l},\pi^\ast_{111l},\lambda^\ast_{00l},\lambda^\ast_{10l},\lambda^\ast_{01l},\lambda^\ast_{11l},\varepsilon^\ast_{0l},\varepsilon^\ast_{1l},\delta^\ast_{l}),\end{align*}
other than that $\theta_l,\theta_l^\ast\in(0,1)^{15}$, it can be shown that the maximum likelihood estimator based on the internal validation design is invariant to its parameterisation (sensitivities/specificities versus positive and negative predictive values). This is because there exists a function mapping every $\theta_l\in(0,1)^{15}$ to a unique $\theta^\ast_l\in(0,1)^{15}$ and vice versa.
Maximising $\ell(\theta)$ with respect to $\theta$ is then equivalent 
to maximising $\ell(\sigma(\theta^\ast))$ ($=\ell^\ast(\theta^\ast)$) with respect to $\theta^\ast$ for some bijection $\sigma$ such that $\theta=\sigma(\theta^\ast)$; that is,
\begin{align*}
\underset{\theta}{\arg\max}~\ell(\theta) = \sigma\left(\underset{\theta^\ast}{\arg\max}~\ell(\sigma(\theta^\ast))\right).
\end{align*}
If more restrictions are imposed on $\theta$ or $\theta$, e.g., if we assume non-saturated logistic models for the components of $\theta$ and $\theta^\ast$, this equivalence no longer holds and the resulting weight estimates may differ depending on the parameterisation.

\subsection{Application}\label{s:estimation_application}

For the re-infarction data example, we assume validation data are available according to a MAR mechanism characterised by
\begin{align*}
\Pr(R_{Y}=1|R_{A}=s,Z=z,B=b,Y=y,A=a,L=l) &= s,\\
\Pr(R_{A}=1|Z=z,B=b,Y=y,A=a,L=l) &= 0.25+0.10b.
\end{align*}
This mechanism assigns validation data to an individual on either both $Y$ and $A$ (30\% of all individuals) or neither depending on their realisation of $B$, the misclassified version of the exposure variable $A$ (Table~\ref{t:table5}). Tables~S.1~and~S.2 (see Supplementary Web Appendix) give the likelihood contributions for the parameterisation based on predictive values and the closed form maximum likelihood expressions, respectively. Maximum likelihood estimates can also be found by fitting to the data the saturated logistic regression models of $B$ and $Z$ on $L$ and $(B,L)$, respectively, and to the validation subset the fully saturated logistic regression models of $A$ and $Y$ on $(Z,B,L)$ and $(A,Z,B,L)$, respectively. Estimated weights are then obtained by plugging in the maximum likelihood estimates into \eqref{e:weights2}. As in the complete data setting where we assumed the weights to be known, evaluating \eqref{e:or4} then yields an odds ratio of $\widehat{\text{OR}}=\text{OR}\approx0.573$.

\begin{table}
\centering
\caption{Expected cell counts (rounded to integers) for illustrative study setting after misclassification and formation of validation subsets.\break
}\label{t:table5}
\begin{tabular*}{\textwidth}{
@{\extracolsep{\fill}}lllllrrlrr@{\extracolsep{\fill}}
}
\toprule
&&&&&\multicolumn{2}{c}{$B=0$}&&\multicolumn{2}{c}{$B=1$}
\\\cline{6-7}\cline{9-10}
$R_Y$&$R_A$&$Y$&$A$&$L$&$Z=0$&$Z=1$&&$Z=0$&$Z=1$\\\midrule
$0$&$0$&$ $&$ $&$0$&$m_{1}=9371$&$m_{2}=7147$&&$m_{3}=1011$&$m_{4}=\phantom{0}884$\\
$0$&$0$&$ $&$ $&$1$&$m_{5}=1120$&$m_{6}=3165$&&$m_{7}=\phantom{0}\phantom{0}80$&$m_{8}=\phantom{0}221$\\
$0$&$1$&$ $&$ $&$ $&$m_{9}=\phantom{0}\phantom{0}\phantom{0}0$&$m_{10}=\phantom{0}\phantom{0}\phantom{0}0$&&$m_{11}=\phantom{0}\phantom{0}\phantom{0}0$&$m_{12}=\phantom{0}\phantom{0}\phantom{0}0$\\
$1$&$0$&$ $&$ $&$ $&$m_{13}=\phantom{0}\phantom{0}\phantom{0}0$&$m_{14}=\phantom{0}\phantom{0}\phantom{0}0$&&$m_{15}=\phantom{0}\phantom{0}\phantom{0}0$&$m_{16}=\phantom{0}\phantom{0}\phantom{0}0$\\
$1$&$1$&$0$&$0$&$0$&$m_{17}=2728$&$m_{18}=\phantom{0}\phantom{0}38$&&$m_{19}=\phantom{0}144$&$m_{20}=\phantom{0}\phantom{0}\phantom{0}2$\\
$1$&$1$&$1$&$0$&$0$&$m_{21}=\phantom{0}\phantom{0}13$&$m_{22}=\phantom{0}\phantom{0}\phantom{0}3$&&$m_{23}=\phantom{0}169$&$m_{24}=\phantom{0}\phantom{0}53$\\
$1$&$1$&$0$&$1$&$0$&$m_{25}=\phantom{0}382$&$m_{26}=3797$&&$m_{27}=\phantom{0}\phantom{0}12$&$m_{28}=\phantom{0}242$\\
$1$&$1$&$1$&$1$&$0$&$m_{29}=\phantom{0}\phantom{0}\phantom{0}1$&$m_{30}=\phantom{0}\phantom{0}\phantom{0}9$&&$m_{31}=\phantom{0}\phantom{0}12$&$m_{32}=\phantom{0}178$\\
$1$&$1$&$0$&$0$&$1$&$m_{33}=\phantom{0}287$&$m_{34}=\phantom{0}\phantom{0}41$&&$m_{35}=\phantom{0}\phantom{0}\phantom{0}6$&$m_{36}=\phantom{0}\phantom{0}\phantom{0}5$\\
$1$&$1$&$1$&$0$&$1$&$m_{37}=\phantom{0}\phantom{0}\phantom{0}2$&$m_{38}=\phantom{0}\phantom{0}\phantom{0}1$&&$m_{39}=\phantom{0}\phantom{0}\phantom{0}7$&$m_{40}=\phantom{0}\phantom{0}\phantom{0}3$\\
$1$&$1$&$0$&$1$&$1$&$m_{41}=\phantom{0}\phantom{0}84$&$m_{42}=1658$&&$m_{43}=\phantom{0}\phantom{0}10$&$m_{44}=\phantom{0}\phantom{0}87$\\
$1$&$1$&$1$&$1$&$1$&$m_{45}=\phantom{0}\phantom{0}\phantom{0}1$&$m_{46}=\phantom{0}\phantom{0}\phantom{0}4$&&$m_{47}=\phantom{0}\phantom{0}\phantom{0}3$&$m_{48}=\phantom{0}\phantom{0}24$\\
\bottomrule
\end{tabular*}
\end{table}

\section{Simulations}\label{s:simulation}
We performed a series of Monte Carlo simulation experiments to illustrate the implementation of the proposed method, to study its finite sample properties and to compare the method to estimators that ignore the presence of confounding or joint exposure and outcome misclassification. All simulations were conducted using R-3.5.0 \citep{R2018} on x86\_64-pc-linux-gnu platforms of the high performance computer cluster of Leiden University Medical Center.

\subsection{Methods}
For all 36 simulation experiments, we generated $n_{\mathrm{sim}}=1000$ samples of size $n$ according to the data generating mechanisms depicted in the directed acyclic graphs of Figure~\ref{f:DAG}. This multi-step data generating process included generating values on measurement error-free variables, introducing misclassification and allocating individuals validation data. We applied various estimators to each of the simulation samples to yield, for each scenario, an empirical distribution of each point estimator and corresponding precision estimators. These distributions were then summarised into various performance metrics. These metrics include the empirical bias of the estimator on the log-scale (i.e., the mean estimated log-OR minus the target log-OR across the $n_{\mathrm{sim}}$ samples), the empirical standard error (SE) of the estimator on the log-scale (i.e., the square root of the mean squared deviation of the estimated log-OR from the mean log-OR), the empirical mean squared error (MSE) (i.e., the sum of the squared SE and the squared bias), the square root of the mean estimated variance (SSE, sample standard error) and the empirical coverage probability (CP) (i.e., the fraction of simulation runs per scenario where the 95\% confidence interval (95\%CI) contained the target quantity).

\subsubsection{Distribution of measurement error-free variables}

Following \cite{Gravel2018}, we consider a setting based on that of ``Scenario A'' in \cite{Setoguchi2008} with slight modifications to the propensity score and outcome models. We consider a fully observed covariate vector $L=(L_0,...,L_{10})$ whose distribution coincides with that of $h(V)$, where $V=(V_1,...,V_{10})$ has the multivariate normal distribution with zero means, unit variances and correlations equal to zero except for the correlations between $W_1$ and $V_5$, $V_2$ and $V_6$, $V_3$ and $V_8$, and $V_4$ and $V_9$, which were set to $0.2$, $0.9$, $0.2$, and $0.9$, respectively. Function $h$ was defined such that $$h(V)=(I(V_1>0),V_2,I(V_3>0),V_4,I(V_5>0),I(V_6>0),V_7,I(V_8>0),I(V_9>0),V_{10}).$$ Thus, sampling from the distribution of $L$ is equivalent to sampling from the multivariate normal distribution with the given parameter values and dichotomising the 1st, 3rd, 5th, 6th, 8th and 9th elements. 

Next, let $U_1$ and $U_2$ be binary variables distributed according to the following logistic models:
\begin{align}
\logit \Pr(U_1=1|L) &= \eta_0, \label{e:B}\\
\logit \Pr(U_2=1|L,U_1) &= \mu_0. \label{e:Z}
\end{align}
The distribution of the binary exposure variable $A$ was defined according to the model \begin{align}
\logit \Pr(A=1|L,U_1,U_2) = \alpha_0+\textstyle\sum_{j=1}^{10}\alpha_jL_j+\alpha_{11}U_1. \label{e:A}
\end{align}
Letting $U_3$ be a scalar random variable that is independent of $(A,L_1,...,L_{10},U_1,U_2)$ and uniformly distributed over the interval $[0,1]$, we defined the counterfactual outcome $Y(a)$, under the intervention setting $A$ to $a$, as
\begin{align}
Y(a) = I\Big(U_3<\expit\Big\{\beta_0+\gamma a+\textstyle\sum_{j=1}^{10}\beta_jL_j+\beta_{11}U_2\Big\}\Big). \label{e:Y}
\end{align}
With $Y:=Y(A)$, the above implies consistency, conditional exchangeability given $L$ and structural positivity.

\subsubsection{Misclassification mechanism}

For scenarios with joint misclassification, we defined $B=U_1$ and $Z=U_2$, so that the predictive values take a standard logistic form:
\begin{align}
\logit \Pr(Y=1|A,B,L,Z) &= \beta_0+\gamma A+\textstyle\sum_{j=1}^{10}\beta_jL_j+\beta_{11}Z
\label{e:pv1}
\\
\logit \Pr(A=1|B,L,Z) &= \alpha_0+\textstyle\sum_{j=1}^{10}\alpha_jL_j+\alpha_{11}B. \label{e:pv2}
\end{align}
For scenarios without exposure misclassification, we set $\alpha_{11}=0$ and defined $B=A$ and $Z=U_2$, so that
\begin{align}
\logit \Pr(Y=1|A,B,L,Z) &= \beta_0+\gamma A+\textstyle\sum_{j=1}^{10}\beta_jL_j+\beta_{11}Z
\label{e:pv3}
\\
\logit \Pr(B=1|L,Z) &= \alpha_0+\textstyle\sum_{j=1}^{10}\alpha_jL_j. \label{pv4}
\end{align}

For simplicity, we removed any marginal dependence of $Z$ on the covariates $L$ and $U_1$ as well as any marginal dependence of $U_1$ on $L$ (cf. equations \eqref{e:B} and \eqref{e:Z}). 
Although models \eqref{e:B} through \eqref{e:pv2} take a standard logistic form, they do not imply that the corresponding sensitivities and specificities can be written in the same form. We chose the predictive values rather than the sensitivities and specificities to take a standard logistic form so as to ensure correct model specification in the estimation of the weights in the simulation experiments, in which a likelihood approach based on predictive values was adopted (cf. \eqref{e:loglikpredval}).

\begin{figure}\centering
\includegraphics[scale=.6]{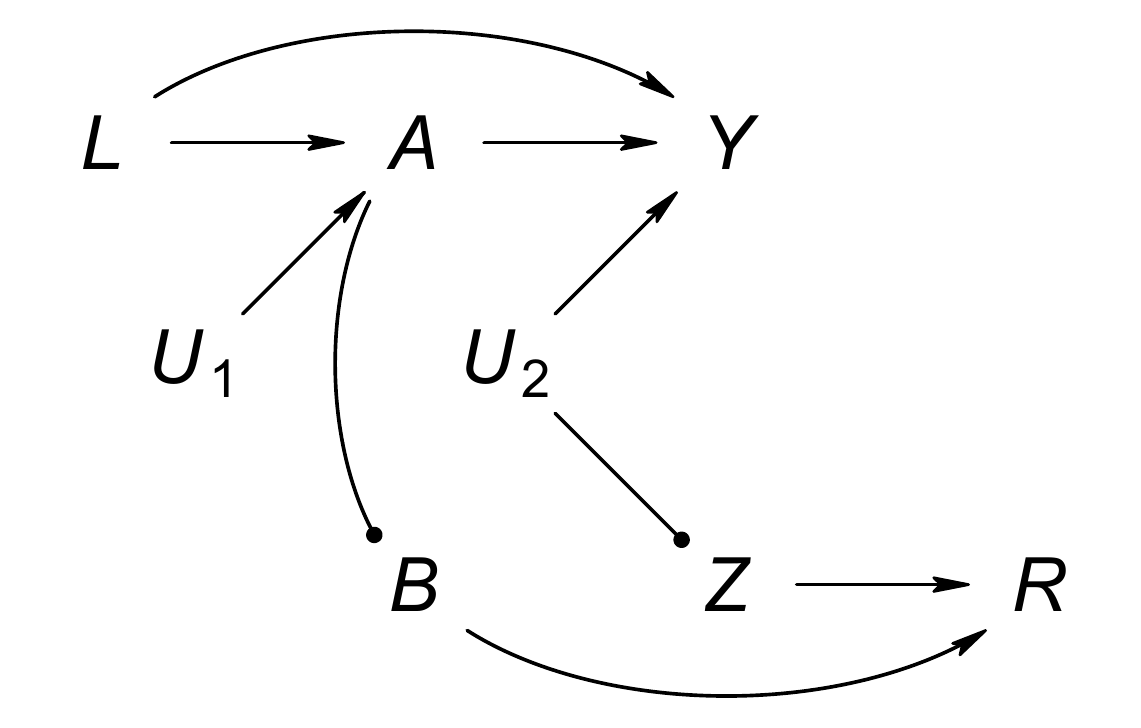}
\includegraphics[scale=.6]{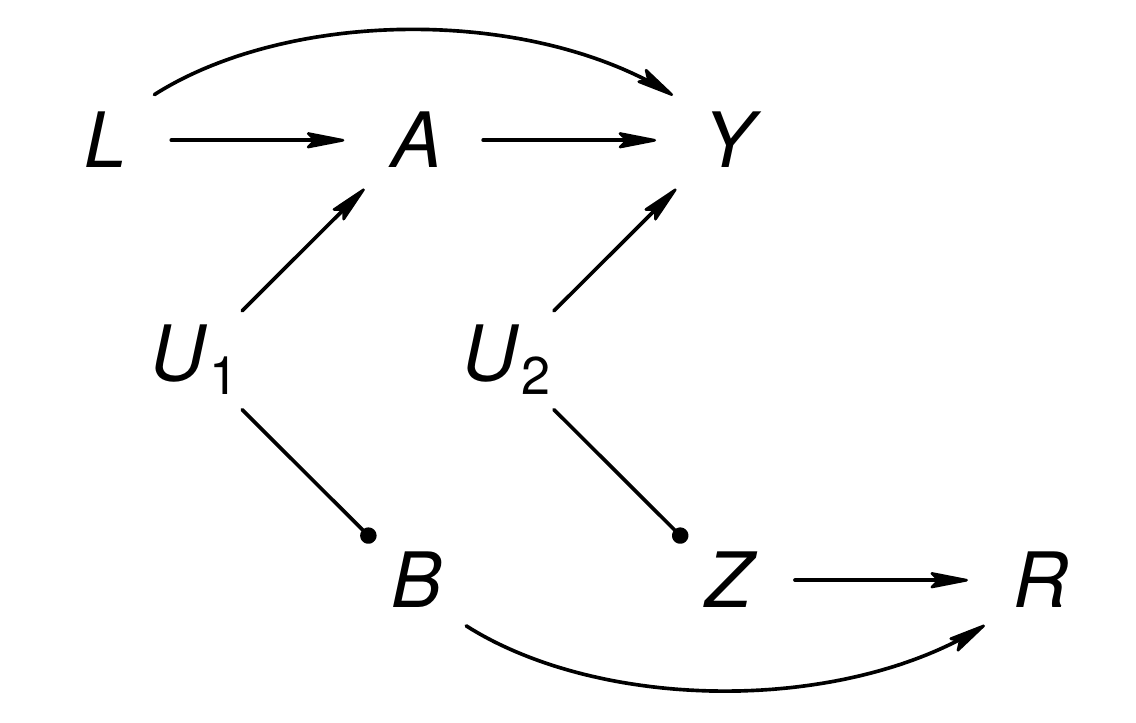}
\caption{Data structure for scenarios with misclassification on the outcome only (left) or on both the exposure and outcome (right). Bullet arrowheads represent deterministic relationships.}\label{f:DAG}
\end{figure}

\subsubsection{Missing data mechanism} 

For these simulations, we stipulated $L$, $B$ and $Z$ to be observed for all subjects. We consider scenarios where the dataset can be partitioned into a subset with validation data on all misclassified variables (denoted $R=1$) and a dataset with validation data on neither ($R=0$). That is, we simulated data such that subjects have validation data on both $A$ and $Y$ or neither on $A$ nor on $Y$. Values for the response indicator $R$ were generated according to the following (MAR) model:
\begin{align}
\logit\Pr(R=1|Z,B,Y,A,L) &= \logit\Pr(R=1|Z,B,L) \nonumber\\
&= \xi_0+\xi_1Z+\xi_2B+\xi_3ZB. \nonumber
\end{align}

\subsubsection{Scenarios}

We initially fixed most parameters of models \eqref{e:A} and \eqref{e:Y} at the respective values of ``Scenario A'' of \cite{Setoguchi2008}: $\alpha_1=0.8$, $\alpha_2=-0.25$, $\alpha_3=0.6$, $\alpha_4=-0.4$, $\alpha_5=-0.8$, $\alpha_6=-0.5$, $\alpha_7=0.7$, $\alpha_8=0$, $\alpha_9=0$, $\alpha_{10}=0$, $\beta_0=-3.85$, $\beta_1=0.3$, $\beta_2=-0.36$, $\beta_3=-0.73$, $\beta_4=-0.2$, $\beta_5=0$, $\beta_6=0$, $\beta_7=0$, $\beta_8=0.71$, $\beta_9=-0.19$ and $\beta_{10}=0.26$. 
Parameters $\eta_0$ and $\alpha_0$ were fixed at zero and $\xi_1$, $\xi_2$ and $\xi_3$ at $2$, $1$ and $-1$, respectively.
The remaining parameters and $\beta_{0}$ were allowed to vary across scenarios as per Table~\ref{t:parVal}.

Scenarios differ by sample size $n$, the presence of outcome misclassification, potentially misclassified outcome prevalence (via $\mu_0$), the associations between the exposure and outcome on the one hand and the respective misclassified versions on the other (via $\alpha_{11}$ and $\beta_{11}$), outcome model intercept $\beta_{0}$, the conditional log-OR $\gamma$, or the size of the validation subset (via $\xi_0$). Based on an iterative Monte Carlo integration approach \citep{Austin2008}, we specified $\gamma$ so as to keep the target marginal log odds ratio at $-0.4$.

\begin{table}
\centering
\caption{Simulation parameter values used in the Monte Carlo studies.\break
}\label{t:parVal}
\begin{tabular*}{\textwidth}{
@{\extracolsep{\fill}}lllcccccccc@{\extracolsep{\fill}}
}
\toprule
&Exposure&&&&&&&
\\
Scenario&misclassification&$n$&$\mu_{0}$&$\alpha_{11}$&$\beta_0$&$\beta_{11}$&$\gamma$&$\xi_0$\\\midrule
1&Absent&$5000$&$-2$&$0$&$-3.85$&$2$&$-0.431$&$-1.5$\\
2&Absent&$5000$&$-3$&$0$&$-3.85$&$2$&$-0.417$&$-1.5$\\
3&Absent&$5000$&$-2$&$0$&$-3.85$&$4$&$-0.624$&$-1.5$\\
4&Absent&$5000$&$-2$&$0$&$-3.85$&$2$&$-0.431$&$-2.5$\\
5&Present&$5000$&$-2$&$2$&$-3.85$&$2$&$-0.431$&$-1.5$\\
6&Present&$5000$&$-3$&$2$&$-3.85$&$2$&$-0.417$&$-1.5$\\
7&Present&$5000$&$-2$&$4$&$-3.85$&$2$&$-0.431$&$-1.5$\\
8&Present&$5000$&$-2$&$2$&$-3.85$&$4$&$-0.624$&$-1.5$\\
9&Present&$5000$&$-2$&$2$&$-3.85$&$2$&$-0.431$&$-2.5$\\
10&Absent&$10000$&$-2$&$0$&$-3.85$&$2$&$-0.431$&$-1.5$\\
11&Absent&$10000$&$-3$&$0$&$-3.85$&$2$&$-0.417$&$-1.5$\\
12&Absent&$10000$&$-2$&$0$&$-3.85$&$4$&$-0.624$&$-1.5$\\
13&Absent&$10000$&$-2$&$0$&$-3.85$&$2$&$-0.431$&$-2.5$\\
14&Present&$10000$&$-2$&$2$&$-3.85$&$2$&$-0.431$&$-1.5$\\
15&Present&$10000$&$-3$&$2$&$-3.85$&$2$&$-0.417$&$-1.5$\\
16&Present&$10000$&$-2$&$4$&$-3.85$&$2$&$-0.431$&$-1.5$\\
17&Present&$10000$&$-2$&$2$&$-3.85$&$4$&$-0.624$&$-1.5$\\
18&Present&$10000$&$-2$&$2$&$-3.85$&$2$&$-0.431$&$-2.5$\\
19&Absent&$5000$&$-2$&$0$&$-2$&$2$&$-0.470$&$-1.5$\\
20&Absent&$5000$&$-3$&$0$&$-2$&$2$&$-0.445$&$-1.5$\\
21&Absent&$5000$&$-2$&$0$&$-2$&$4$&$-0.641$&$-1.5$\\
22&Absent&$5000$&$-2$&$0$&$-2$&$2$&$-0.470$&$-2.5$\\
23&Present&$5000$&$-2$&$2$&$-2$&$2$&$-0.470$&$-1.5$\\
24&Present&$5000$&$-3$&$2$&$-2$&$2$&$-0.445$&$-1.5$\\
25&Present&$5000$&$-2$&$4$&$-2$&$2$&$-0.470$&$-1.5$\\
26&Present&$5000$&$-2$&$2$&$-2$&$4$&$-0.641$&$-1.5$\\
27&Present&$5000$&$-2$&$2$&$-2$&$2$&$-0.470$&$-2.5$\\
28&Absent&$10000$&$-2$&$0$&$-2$&$2$&$-0.470$&$-1.5$\\
29&Absent&$10000$&$-3$&$0$&$-2$&$2$&$-0.445$&$-1.5$\\
30&Absent&$10000$&$-2$&$0$&$-2$&$4$&$-0.641$&$-1.5$\\
31&Absent&$10000$&$-2$&$0$&$-2$&$2$&$-0.470$&$-2.5$\\
32&Present&$10000$&$-2$&$2$&$-2$&$2$&$-0.470$&$-1.5$\\
33&Present&$10000$&$-3$&$2$&$-2$&$2$&$-0.445$&$-1.5$\\
34&Present&$10000$&$-2$&$4$&$-2$&$2$&$-0.470$&$-1.5$\\
35&Present&$10000$&$-2$&$2$&$-2$&$4$&$-0.641$&$-1.5$\\
36&Present&$10000$&$-2$&$2$&$-2$&$2$&$-0.470$&$-2.5$\\
\bottomrule
\end{tabular*}
\end{table}

\subsubsection{Estimators}
We considered five estimators of the OR for the marginal exposure-outcome effect: a crude estimator (labeled Crude) that ignores both confounding and misclassication of any variable, a misclassification naive estimator (labeled PS) that addresses confounding through IPW, complete cases analysis (CCA) in which IPW is applied only to the subset of subjects with validation data, the Gravel and Platt estimator (GP) that ignores exposure misclassification, and the method proposed in this article (labeled IPWM). 
Both GP and IPWM are implemented using the R function \texttt{mecor::ipwm} \citep{mecor,Nab2018}, which in the simulation settings considered uses iteratively reweighted least squares via the \texttt{stats::glm} function for maximum likelihood estimation.
Unlike Gravel and Platt \citep{Gravel2018}, we used a non-parametric rather than a semi-parametric bootstrap procedure for estimating standard errors and constructing confidence intervals. Semi-parametrically generating response indicators would preferably require modelling of (or making additional assumptions about) the missing data mechanism. 
For all methods and each original dataset, we drew 1000 bootstrap samples for variance estimation and the construction of percentile confidence intervals.

All estimators are based on a function of the estimated outcome probability $P_1$ in the exposed group and the estimated outcome probability $P_0$ in the unexposed group. However, since $P_1$ and $P_0$ may take a value of 0 or 1, the crude odds ratio $[P_1/(1-P_1)]/[P_0/(1-P_0)]$ need not exist. 
In contrast to what is often (implicitly) done in simulation studies---i.e., studying the properties of the estimators after conditioning on datasets where $[P_1/(1-P_1)]/[P_0/(1-P_0)]$ is defined---we first define $P_1^\ast=(P_1s+1)/(s+2)$ and $P_0^\ast=(P_0s+1)/(s+2)$ for a large positive number $s$ (here set to $10^6$) and then regard $[P_1^\ast/(1-P_1^\ast)]/[P_0^\ast/(1-P_0^\ast)]$ as the estimator of the OR for the exposure-outcome association. This ensures the estimator is always defined and effectively shrinks the outcome probabilities towards 0.5 and the OR towards 1 (Appendix~\nameref{s:app_shrinkage}).

For PS and CCA, we used a logistic regression of $B$ and $A$, respectively, on covariates $L_1$ through $L_{10}$ as main effects to estimate the propensity scores. Taking the crude OR for the association between $B$ and $Z$ (PS) or $A$ and $Y$ (CCA) over the data weighted by the reciprocal of the propensity scores provided an estimate of target OR. R code for the methods GP and IPWM is given in Appendix~\nameref{s:app_Rcode}.

\subsection{Results}\label{s:sim_results}
The treatment assignment mechanism detailed above resulted in average exposure rates ranging from 17\% to 51\%, whereas average outcome rates ranged from 3\% to 22\%.
Across all simulation studies, the average outcome, exposure and joint misclassification rates ranged from 6\% to 18\%, from 0 to 33\% and from 0\% to 6\%, respectively. Approximately 16\% to 32\% of subjects were allocated validation data.

The results on the performance of the various methods in simulations studies 1-9 are provided in Table~\ref{t:sim_results} (see Supplementary Table S.3 for the results on all scenarios).

As expected, Crude, PS and CCA clearly showed bias with respect to the target log OR of $-0.4$. The bias associated with restricting the analysis to records with validation data is likely brought on to a large extent by collider stratification, with $R$ acting as the collider here (cf. Figure~\ref{f:DAG}). Both Crude and PS indicated a null effect, as one would anticipate in view of the marginal and $L$-conditional independence of $B$ and $Z$ implied by the simulation set-up. The empirical coverage probabilities were, although low for both estimators, similar to substantially larger for PS as compared with Crude. Paralleling this is that Crude, whose (implicit) propensity score model is inherently at least as parsimonious, yielded similar to smaller empirical and sample standard errors as compared with PS. With the average fraction of subjects with validation data being as low as 16\% (in scenarios with low $\xi_0$) to 32\%, it is not unsurprising that Crude was subject to the largest degree of variability.

The results for the IPWM approach are generally favourable and in line with its theoretical (large sample) properties. Note that the results for GP and IPWM are identical for scenarios 1-4, 10-13, 19-22 and 28-31, since the methods are equivalent in the absence of exposure misclassification. In all other scenarios, i.e., scenarios for which GP was not developed, GP performed substantially worse than IPWM.
The non-zero, albeit relatively small, systematic deviations of the IPWM point estimates from the target $-0.4$, notably the estimated bias of $-0.097$ (scenario 2), may be attributable in part to the outcome being rare (with prevalence ranging from 3\% to 8\% across scenarios 1-9). This is indicated by the superior performance of IPWM in scenarios where the outcome is more prevalent (scenarios 10-36, prevalence up to 22\%). A similar observation was made by \cite{Gravel2018}.
The standard errors for GP and IPWM were noticeably higher than those of Crude and PS, which is unsurprising in view of the discrepancies in the number of estimated parameters. As expected, increasing the sample size, the true outcome rate (via $\beta_0$) or both led to a decrease in the variability of IPWM (cf. Table~\ref{t:parVal} and Supplementary Table S.3). Throughout the empirical coverage probabilities of IPWM were close to the nominal level of 0.95.

\begin{table}
\centering
\caption{Results for simulation studies 1-9 on the performance of different causal estimators in various scenarios of confounding and misclassification in exposure and outcome. Abbreviations: PS, propensity score method ignoring misclassification; CCA, complete case analysis; GP, Gravel and Platt estimator ignoring exposure misclassification; IPWM, inverse probability weighting method for confounding and joint exposure and outcome misclassification; BSE, estimated standard error for the bias due to Monte Carlo error; SE, empirical standard error; SSE, sample standard error; CP, empirical coverage probability. In all scenarios, the true marginal log OR (estimand) was $-0.4$.\break
}\label{t:sim_results}
\begin{tabular*}{\textwidth}{
@{\extracolsep{\fill}}lcccccc@{\extracolsep{\fill}}
}
\toprule
&\multicolumn{6}{c}{Crude}\\\cline{2-7}
Scenario&Bias&BSE&MSE&SE&SSE&CP
\\\midrule
1 & $\phantom{-}0.394$ & $0.004$ & $0.275$ & $0.119$ & $0.118$ & $0.122$ \\
2 & $\phantom{-}0.382$ & $0.006$ & $0.329$ & $0.183$ & $0.184$ & $0.492$ \\
3 & $\phantom{-}0.394$ & $0.004$ & $0.272$ & $0.117$ & $0.118$ & $0.116$ \\
4 & $\phantom{-}0.401$ & $0.004$ & $0.278$ & $0.117$ & $0.118$ & $0.102$ \\
5 & $\phantom{-}0.401$ & $0.003$ & $0.252$ & $0.090$ & $0.088$ & $0.007$ \\
6 & $\phantom{-}0.407$ & $0.004$ & $0.298$ & $0.132$ & $0.134$ & $0.133$ \\
7 & $\phantom{-}0.396$ & $0.003$ & $0.243$ & $0.086$ & $0.088$ & $0.009$ \\
8 & $\phantom{-}0.395$ & $0.003$ & $0.242$ & $0.086$ & $0.088$ & $0.005$ \\
9 & $\phantom{-}0.398$ & $0.003$ & $0.247$ & $0.089$ & $0.088$ & $0.005$ \\
\midrule
&\multicolumn{6}{c}{PS}\\\cline{2-7}
Scenario&Bias&BSE&MSE&SE&SSE&CP
\\\midrule
1 & $\phantom{-}0.392$ & $0.005$ & $0.321$ & $0.168$ & $0.169$ & $0.382$ \\
2 & $\phantom{-}0.379$ & $0.008$ & $0.407$ & $0.264$ & $0.258$ & $0.738$ \\
3 & $\phantom{-}0.389$ & $0.006$ & $0.327$ & $0.175$ & $0.169$ & $0.402$ \\
4 & $\phantom{-}0.389$ & $0.006$ & $0.327$ & $0.176$ & $0.168$ & $0.392$ \\
5 & $\phantom{-}0.402$ & $0.003$ & $0.252$ & $0.090$ & $0.088$ & $0.010$ \\
6 & $\phantom{-}0.407$ & $0.004$ & $0.297$ & $0.131$ & $0.135$ & $0.136$ \\
7 & $\phantom{-}0.396$ & $0.003$ & $0.243$ & $0.086$ & $0.088$ & $0.009$ \\
8 & $\phantom{-}0.395$ & $0.003$ & $0.242$ & $0.086$ & $0.088$ & $0.004$ \\
9 & $\phantom{-}0.398$ & $0.003$ & $0.247$ & $0.089$ & $0.088$ & $0.005$ \\
\midrule
&\multicolumn{6}{c}{CCA}\\\cline{2-7}
Scenario&Bias&BSE&MSE&SE&SSE&CP
\\\midrule
1 & $-0.078$ & $0.015$ & $0.476$ & $0.469$ & $0.491$ & $0.899$ \\
2 & $-0.117$ & $0.019$ & $0.615$ & $0.601$ & $0.900$ & $0.887$ \\
3 & $-0.020$ & $0.010$ & $0.301$ & $0.301$ & $0.300$ & $0.919$ \\
4 & $-0.093$ & $0.020$ & $0.640$ & $0.631$ & $1.158$ & $0.899$ \\
5 & $-0.145$ & $0.009$ & $0.307$ & $0.286$ & $0.286$ & $0.903$ \\
6 & $-0.109$ & $0.011$ & $0.357$ & $0.345$ & $0.362$ & $0.930$ \\
7 & $-0.213$ & $0.007$ & $0.282$ & $0.237$ & $0.250$ & $0.865$ \\
8 & $-0.209$ & $0.006$ & $0.231$ & $0.187$ & $0.186$ & $0.775$ \\
9 & $-0.175$ & $0.012$ & $0.422$ & $0.392$ & $0.411$ & $0.902$
\\
\midrule
\end{tabular*}
\end{table}

\begin{table}
\centering
\contcaption{(Continued.)\break
}
\begin{tabular*}{\textwidth}{
@{\extracolsep{\fill}}lcccccc@{\extracolsep{\fill}}
}
\toprule
&\multicolumn{6}{c}{GP}\\\cline{2-7}
Scenario&Bias&BSE&MSE&SE&SSE&CP
\\\midrule
1 & $-0.036$ & $0.011$ & $0.360$ & $0.359$ & $0.428$ & $0.958$ \\
2 & $-0.097$ & $0.016$ & $0.515$ & $0.505$ & $0.861$ & $0.938$ \\
3 & $-0.019$ & $0.007$ & $0.234$ & $0.233$ & $0.240$ & $0.939$ \\
4 & $-0.045$ & $0.016$ & $0.503$ & $0.501$ & $1.087$ & $0.944$ \\
5 & $\phantom{-}0.269$ & $0.008$ & $0.316$ & $0.244$ & $0.244$ & $0.799$ \\
6 & $\phantom{-}0.280$ & $0.010$ & $0.392$ & $0.314$ & $0.339$ & $0.862$ \\
7 & $\phantom{-}0.134$ & $0.008$ & $0.259$ & $0.241$ & $0.252$ & $0.926$ \\
8 & $\phantom{-}0.259$ & $0.004$ & $0.207$ & $0.140$ & $0.144$ & $0.570$ \\
9 & $\phantom{-}0.263$ & $0.010$ & $0.394$ & $0.325$ & $0.339$ & $0.883$ \\
\midrule
&\multicolumn{6}{c}{IPWM}\\\cline{2-7}
Scenario&Bias&BSE&MSE&SE&SSE&CP\\\midrule
1 & $-0.036$ & $0.011$ & $0.360$ & $0.359$ & $0.428$ & $0.958$ \\
2 & $-0.097$ & $0.016$ & $0.515$ & $0.505$ & $0.861$ & $0.938$ \\
3 & $-0.019$ & $0.007$ & $0.234$ & $0.233$ & $0.240$ & $0.939$ \\
4 & $-0.045$ & $0.016$ & $0.503$ & $0.501$ & $1.087$ & $0.944$ \\
5 & $-0.017$ & $0.009$ & $0.286$ & $0.286$ & $0.284$ & $0.942$ \\
6 & $-0.014$ & $0.011$ & $0.359$ & $0.359$ & $0.386$ & $0.958$ \\
7 & $\phantom{-}0.004$ & $0.008$ & $0.243$ & $0.243$ & $0.261$ & $0.969$ \\
8 & $-0.004$ & $0.006$ & $0.180$ & $0.180$ & $0.181$ & $0.958$ \\
9 & $-0.025$ & $0.012$ & $0.375$ & $0.374$ & $0.415$ & $0.956$ \\
\bottomrule
\end{tabular*}
\end{table}

\section{Discussion}\label{s:discussion}
The analysis of epidemiologic data is often complicated by the presence of confounding and misclassifications in exposure and outcome variables. In this paper we proposed a new estimator for estimating a marginal odds-ratio in the presence of confouding and joint misclassification of the exposure and outcome variables. In simulation studies, this weighting estimator showed promising finite sample performance, reducing bias and mean squared error as compared with simpler methods.

The proposed IPWM estimator is an extension of the inverse probability weighting estimator recently proposed by \cite{Gravel2018} (GP) which only addresses the misclassification in the outcome. IPWM and GP are (mathematically) equivalent when the exposure is (assumed to be) measured without misclassification error.

Like the Gravel and Platt approach, IPWM relies on estimates of sensitivity and specificity or positive and negative predictive values for the misclassified variables. In this paper, we used an internal approach where a portion of subjects would receive error-free (`gold standard') measurements on either or both the outcome and exposure. 

We anticipate that in some settings the likelihood may not be fully identifiable from the data at hand. In these settings, it may be possible to incorporate external rather than internal information on the misclassification rates, possibly through a Bayesian approach using prior assumptions about misclassification probabilities. When validation data is external, however, it may be necessary to assume misclassification to be independent of covariates $L$, because external studies seldom consider the same covariates as the main study \citep{Lyles2011}. External validation approaches also require the assumption that the misclassification parameters targeted in the validation sample are transportable to the main study. 

An important advantage of the IPWM approach is that the subjects with validation data need not form a completely random subset. The proposed method was developed under the assumption that validation data allocation occurs in an `ignorable' fashion \citep{Rubin1976}. In practice, it may be that the researchers have limited control over the validation data allocation mechanism. For instance, it is conceivable that individuals with specific indications (e.g., with a realisation of $L$, $B$ or $Z$) are practically ineligible to be assigned a double measurement of the exposure ($A$ and $B$) and outcome ($Y$ and $Z$). Further, the estimator also allows for validation subjects to receive either the double exposure or double outcome measurement. We simulated data such that subjects have validation data on both the exposure and outcome variables or on neither. Although this may greatly simplify analysis and enhance efficiency, in practice it is not necessary to assume that this condition holds. An interesting scenario is where subjects have validation data on at most one variable, i.e., on the exposure variable or the outcome variable but not both. In this case, valid estimation would require additional modelling assumptions; for example, the  error-free outcome variable cannot then be regressed on the error-free exposure variable.

To accommodate settings where validation data allocation is not completely at random, we deviated from the semi-parametric bootstrap procedure for variance estimation proposed by Gravel and Platt. Instead, the non-parametric procedure we used requires less assumptions regarding the validation subset sampling procedure. The non-parametric procedure showed good performance in our simulations.



Whilst we have discussed under what conditions the proposed method consistently estimates or at least identifies the target quantity, the assumptions 
may be untenable in particular settings. Particularly, an infallible measurement tool for the exposure and outcome that can be performed on a subset of the data need not always exist. The robustness to deviations of infallibility is an interesting and important direction for further research. This is especially relevant where there exists considerable uncertainty about the tenability of the assumptions that is difficult to incorporate in the analysis. An obvious and flexible alternative to IPWM is to multiply impute missing values including absent measurement error-free variables before implementing IPW (MI+IPW). Although MI+IPW and IPWM may be comparable in terms of their assumptions, it is yet unclear how they behave under assumption violations such as misspecification of the outcome model. 

An advantageous property of MI+IPW is that it can easily accommodate missing covariate values. Other alternatives that can accommodate missing covariates were recently developed by \cite{Shu2018}. Their proposed weighting estimators simultaneously addresses confounding, misclassification of the outcome (but not of the exposure) and measurement error on the covariates under a classical additive measurement error model. The methods can be implemented using validation data or repeated measurements and use a simple misclassification model (in which the outcome surrogate is independent of exposure or covariates given the target outcome) that is suitable for performing sensitivity analyses.

Another interesting area for further research is where the researchers do have control over who is referred for further testing by the assumed infallible measurement tool(s). An obvious choice is to adopt a completely at random strategy (simple random sampling). However, other referral (sampling) strategies exist and it is not clear what strategy leads to the most favourable estimator properties for the given setting.  

In summary, we have developed an extension to an existing method, to allow for valid estimation of a marginal causal OR in the presence of confounding and a commonly ignored and misunderstood source of bias---joint exposure and outcome misclassification. The R function \texttt{mecor::ipwm} has been made available to facilitate implementation \citep{mecor,Nab2018}.


\backmatter

\section*{Acknowledgements}
RHHG was funded by the Netherlands Organization for Scientific Research (NWO-Vidi project 917.16.430). The views expressed in this article are those of the authors and not necessarily any funding body.


\bibliographystyle{biom} 

\begin{thebibliography}{}

\bibitem[\protect\citeauthoryear{Austin and Stafford}{Austin and
  Stafford}{2008}]{Austin2008}
Austin, P.~C. and Stafford, J. (2008).
\newblock The performance of two data-generation processes for data with
  specified marginal treatment odds ratios.
\newblock {\em Communications in Statistics - Simulation and Computation} {\bf
  37,} 1039--1051.

\bibitem[\protect\citeauthoryear{Brakenhoff, Mitroiu, Keogh, Moons, Groenwold,
  and van Smeden}{Brakenhoff et~al.}{2018}]{Brakenhoff2018}
Brakenhoff, T.~B., Mitroiu, M., Keogh, R.~H., Moons, K.~G., Groenwold, R.~H.,
  and van Smeden, M. (2018).
\newblock Measurement error is often neglected in medical literature: a
  systematic review.
\newblock {\em Journal of Clinical Epidemiology} {\bf 98,} 89--97.

\bibitem[\protect\citeauthoryear{Brenner, Savitz, and Gefeller}{Brenner
  et~al.}{1993}]{Brenner1993}
Brenner, H., Savitz, D.~A., and Gefeller, O. (1993).
\newblock The effects of joint misclassification of exposure and disease on
  epidemiologic measures of association exposure and disease on epidemiologic
  measures of association.
\newblock {\em Journal of Clinical Epidemiology} {\bf 46,} 1195--1202.

\bibitem[\protect\citeauthoryear{Brooks, Getz, Brennan, Pollack, and
  Fox}{Brooks et~al.}{2018}]{Brooks2018}
Brooks, D.~R., Getz, K.~D., Brennan, A.~T., Pollack, A.~Z., and Fox, M.~P.
  (2018).
\newblock The impact of joint misclassification of exposures and outcomes on
  the results of epidemiologic research.
\newblock {\em Current Epidemiology Reports} {\bf 5,} 166--174.

\bibitem[\protect\citeauthoryear{Culver, Ockene, Balasubramanian, Olendzki,
  Sepavich, Wactawski-Wende, Manson, Qiao, Liu, Merriam, et~al\mbox{.}}{Culver
  et~al.}{2012}]{Culver2012}
Culver, A.~L., Ockene, I.~S., Balasubramanian, R., Olendzki, B.~C., Sepavich,
  D.~M., Wactawski-Wende, J., Manson, J.~E., Qiao, Y., Liu, S., Merriam, P.~A.,
  et~al. (2012).
\newblock Statin use and risk of diabetes mellitus in postmenopausal women in
  the women's health initiative.
\newblock {\em Archives of internal medicine} {\bf 172,} 144--152.

\bibitem[\protect\citeauthoryear{Dawid}{Dawid}{1979}]{Dawid1979}
Dawid, A. (1979).
\newblock Conditional independence in statistical theory.
\newblock {\em Journal of the Royal Statistical Society, Series B
  (Methodological)} pages 1--31.

\bibitem[\protect\citeauthoryear{Gravel and Platt}{Gravel and
  Platt}{2018}]{Gravel2018}
Gravel, C.~A. and Platt, R.~W. (2018).
\newblock Weighted estimation for confounded binary outcomes subject to
  misclassification.
\newblock {\em Statistics in Medicine} {\bf 37,} 425--436.

\bibitem[\protect\citeauthoryear{Holland}{Holland}{1986}]{Holland1986}
Holland, P. (1986).
\newblock Statistics in causal inference.
\newblock {\em Journal of the American Statistical Association} {\bf 81,}
  945--960.

\bibitem[\protect\citeauthoryear{Holland}{Holland}{1988}]{Holland1988}
Holland, P. (1988).
\newblock Causal inference, path analysis, and recursive structural equations
  models.
\newblock {\em Sociological Methodology} {\bf 18,} 449--484.

\bibitem[\protect\citeauthoryear{Jurek, Greenland, and Maldonado}{Jurek
  et~al.}{2008}]{Jurek2008}
Jurek, A.~M., Greenland, S., and Maldonado, G. (2008).
\newblock Brief report: how far from non-differential does exposure or disease
  misclassification have to be to bias measures of association away from the
  null?
\newblock {\em International Journal of Epidemiology} {\bf 37,} 382--385.

\bibitem[\protect\citeauthoryear{Jurek, Maldonado, Greenland, and Church}{Jurek
  et~al.}{2006}]{Jurek2006}
Jurek, A.~M., Maldonado, G., Greenland, S., and Church, T.~R. (2006).
\newblock Exposure-measurement error is frequently ignored when interpreting
  epidemiologic study results.
\newblock {\em European journal of epidemiology} {\bf 21,} 871--876.

\bibitem[\protect\citeauthoryear{Kristensen}{Kristensen}{1992}]{Kristensen1992}
Kristensen, P. (1992).
\newblock Bias from nondifferential but dependent misclassification of exposure
  and outcome.
\newblock {\em Epidemiology} pages 210--215.

\bibitem[\protect\citeauthoryear{Leong, Dasgupta, Bernatsky, Lacaille,
  Avina-Zubieta, and Rahme}{Leong et~al.}{2013}]{Leong2013}
Leong, A., Dasgupta, K., Bernatsky, S., Lacaille, D., Avina-Zubieta, A., and
  Rahme, E. (2013).
\newblock Systematic review and meta-analysis of validation studies on a
  diabetes case definition from health administrative records.
\newblock {\em PloS one} {\bf 8,} e75256.

\bibitem[\protect\citeauthoryear{Lyles, Tang, Superak, King, Celentano, Lo, and
  Sobel}{Lyles et~al.}{2011}]{Lyles2011}
Lyles, R.~H., Tang, L., Superak, H.~M., King, C.~C., Celentano, D.~D., Lo, Y.,
  and Sobel, J.~D. (2011).
\newblock Validation data-based adjustments for outcome misclassification in
  logistic regression: an illustration.
\newblock {\em Epidemiology} {\bf 22,} 589.

\bibitem[\protect\citeauthoryear{Marcum, Sevick, and Handler}{Marcum
  et~al.}{2013}]{Marcum2013}
Marcum, Z.~A., Sevick, M.~A., and Handler, S.~M. (2013).
\newblock Medication nonadherence: a diagnosable and treatable medical
  condition.
\newblock {\em Jama} {\bf 309,} 2105--2106.

\bibitem[\protect\citeauthoryear{Nab}{Nab}{2019}]{mecor}
Nab, L. (2019).
\newblock {\em mecor: Measurement Error Corrections}.
\newblock R package version 0.1.0. Available from:
  https://github.com/LindaNab/mecor.git.

\bibitem[\protect\citeauthoryear{Nab, Groenwold, Welsing, and van Smeden}{Nab
  et~al.}{2018}]{Nab2018}
Nab, L., Groenwold, R.~H., Welsing, P.~M., and van Smeden, M. (2018).
\newblock Measurement error in continuous endpoints in randomised trials:
  problems and solutions.
\newblock {\em arXiv preprint arXiv:1809.07068} .

\bibitem[\protect\citeauthoryear{Neyman, Iwaszkiewicz, and {St.
  Kolodziejczyk}}{Neyman et~al.}{1935}]{Neyman1935}
Neyman, J., Iwaszkiewicz, K., and {St. Kolodziejczyk} (1935).
\newblock Statistical problems in agricultural experimentation.
\newblock {\em Supplement to the Journal of the Royal Statistical Society} {\bf
  2,} 107--180.

\bibitem[\protect\citeauthoryear{Ni, Leong, Dasgupta, and Rahme}{Ni
  et~al.}{2017}]{Ni2017}
Ni, J., Leong, A., Dasgupta, K., and Rahme, E. (2017).
\newblock Correcting hazard ratio estimates for outcome misclassification using
  multiple imputation with internal validation data.
\newblock {\em Pharmacoepidemiology and drug safety} {\bf 26,} 925--934.

\bibitem[\protect\citeauthoryear{Pearl}{Pearl}{2009}]{Pearl2009}
Pearl, J. (2009).
\newblock {\em Causality: Models, Reasoning and Inference}.
\newblock Cambridge University Press, New York.

\bibitem[\protect\citeauthoryear{{R Core Team}}{{R Core Team}}{2018}]{R2018}
{R Core Team} (2018).
\newblock {\em R: A Language and Environment for Statistical Computing}.
\newblock R Foundation for Statistical Computing, Vienna, Austria.

\bibitem[\protect\citeauthoryear{Rubin}{Rubin}{1974}]{Rubin1974}
Rubin, D. (1974).
\newblock Estimating causal effects of treatments in randomized and
  nonrandomized studies.
\newblock {\em Journal of Educational Psychology} {\bf 66,} 688--701.

\bibitem[\protect\citeauthoryear{Rubin}{Rubin}{1976}]{Rubin1976}
Rubin, D. (1976).
\newblock Inference and missing data.
\newblock {\em Biometrika} {\bf 63,} 581--592.

\bibitem[\protect\citeauthoryear{Setoguchi, Schneeweiss, MA, Glynn, and
  Cook}{Setoguchi et~al.}{2008}]{Setoguchi2008}
Setoguchi, S., Schneeweiss, S., MA, M.~B., Glynn, R., and Cook, E. (2008).
\newblock Evaluating uses of data mining techniques in propensity score
  estimation: a simulation study.
\newblock {\em Pharmacoepidemiology and Drug Safety} {\bf 17,} 546--555.

\bibitem[\protect\citeauthoryear{Shu and Yi}{Shu and Yi}{2018}]{Shu2018}
Shu, D. and Yi, G.~Y. (2018).
\newblock Weighted causal inference methods with mismeasured covariates and
  misclassified outcomes.
\newblock {\em Statistics in Medicine} .

\bibitem[\protect\citeauthoryear{Tang, Lyles, Ye, Lo, and King}{Tang
  et~al.}{2013}]{Tang2013}
Tang, L., Lyles, R.~H., Ye, Y., Lo, Y., and King, C.~C. (2013).
\newblock Extended matrix and inverse matrix methods utilizing internal
  validation data when both disease and exposure status are misclassified.
\newblock {\em Epidemiologic Methods} {\bf 2,} 49--66.

\bibitem[\protect\citeauthoryear{VanderWeele and Hern{\'a}n}{VanderWeele and
  Hern{\'a}n}{2012}]{VanderWeele2012}
VanderWeele, T.~J. and Hern{\'a}n, M.~A. (2012).
\newblock Results on differential and dependent measurement error of the
  exposure and the outcome using signed directed acyclic graphs.
\newblock {\em American journal of epidemiology} {\bf 175,} 1303--1310.

\bibitem[\protect\citeauthoryear{Vogel, Brenner, Pfahlberg, and Gefeller}{Vogel
  et~al.}{2005}]{Vogel2005}
Vogel, C., Brenner, H., Pfahlberg, A., and Gefeller, O. (2005).
\newblock The effects of joint misclassification of exposure and disease on the
  attributable risk.
\newblock {\em Statistics in Medicine} {\bf 24,} 1881--1896.

\end{thebibliography}

\appendix
\setcounter{equation}{0}
\renewcommand\theequation{\thetheorem.\arabic{equation}}

\newpage
\section{I}\label{s:app_weighting}
\setcounter{theorem}{0}
\renewcommand*{\thetheorem}{A.\arabic{theorem}}

\begin{theorem}\label{t:A1}
For any $a,l$, let
\begin{align*}\varphi(a,l)&=\frac{\varphi^\ast(a,l)}{\Exp[\varphi^\ast(A,L)|A=a]}~~\text{and}~~\varphi^\ast(a,l)=\frac{1}{\Pr(A=a|L=l)}.\end{align*}
If $Y(A)=Y$ (consistency), $(Y(0),Y(1))\CI A|L=l$ (conditional exchangeability), $\Pr(A=a)>0$ and $\Pr(A=a|L=l)>0$ (positivity) for all $a$ and every $l$ in the support of $L$, then 
\begin{align*}
\Exp[Y(a)]&=\Exp[\varphi(A,L)I(Y=1)|A=a].
\end{align*}
\end{theorem}
\begin{proof}
We begin by considering $\Exp[\varphi^\ast(A,L)|A=a]$. By the law of the unconscious statistician and Bayes' theorem, we have
\begin{align*}
\Exp[\varphi^\ast(A,L)|A=a]
	&= \sum_l\frac{\Pr(L=l|A=a)}{\Pr(A=a|L=l)} \\
    &= \sum_l\frac{\Pr(A=a|L=l)\Pr(L=l)}{\Pr(A=a)\Pr(A=a|L=l)} \\
    &= \frac{1}{\Pr(A=a)}\sum_l\Pr(L=l) \\
    &= \frac{1}{\Pr(A=a)}.
\end{align*}
Hence, for all $a,y$, we have
\begin{align}
\sum_l\varphi(a,l)\Pr(Y=y,L=l|A=a) 
	&= \sum_l\frac{\Pr(Y=y,L=l|A=a)\Pr(A=a)}{\Pr(A=a|L=l)} \nonumber\\
    &= \sum_l\frac{\Pr(Y=y|A=a,L=l)\Pr(A=a|L=l)\Pr(L=l)}{\Pr(A=a|L=l)} \nonumber\\
    &= \sum_l\Pr(Y=y|A=a,L=l)\Pr(L=l) \nonumber\\
    &= \sum_l\Pr(Y(a)=y|A=a,L=l)\Pr(L=l) \label{Consistency}\\
    &= \sum_l\Pr(Y(a)=y|L=l)\Pr(L=l) \label{ConditionalExchangeability}\\
    &= \Pr(Y(a)=y), \nonumber
\end{align}
where \eqref{Consistency} and \eqref{ConditionalExchangeability} hold under consistency and conditional exchangeability given $L$, respectively. Positivity ensures the weights are defined/exist. Hence, $\Exp[\varphi(A,L)I(Y=1)|A=a]=\sum_l\varphi(a,l)\Pr(Y=1,L=l|A=a)=\Exp[Y(a)]$, as desired.
\end{proof}

\setcounter{appendixcorollary}{0}
\renewcommand*{\theappendixcorollary}{A.\arabic{appendixcorollary}}

\begin{appendixcorollary} \label{TC1}
For any $y,a,l$, let
\begin{align*}\varphi(a,l)&=\frac{\varphi^\ast(a,l)}{\Exp[\varphi^\ast(A,L)|A=a]},~~\varphi^\ast(a,l)=\frac{1}{\Pr(A=a|L=l)},~~\text{and}\\
\phi(a,l)&=\frac{\Pr(Y=1,L=l|A=a)}{\Pr(Z=1,L=l|B=a)}.
\end{align*}
If $Y(A)=Y$, $(Y(0),Y(1))\CI A|L$ and positivity holds, then 
\begin{align*}
\Exp[Y(a)] &= \sum_l\varphi(a,l)\Pr(Y=1,L=l|A=a)\\
&= \sum_l\varphi(a,l)\phi(a,l)\Pr(Z=1,L=l|B=a)\\
&=
\Exp[\varphi(B,L)\phi(B,L)Z|B=a].
\end{align*}
\end{appendixcorollary}

\section{II}\label{s:app_shrinkage}
\begin{theorem}
Fix some $s>0$ and let $P^\ast=(Ps+1)/(s+2)$ for all $P\in[0,1]$. If $(P_0,P_1)\in (0,1)\times(0,1)$, then 
\begin{align*}
1<\frac{P_1^\ast/(1-P_1^\ast)}{P_0^\ast/(1-P_0^\ast)} < \frac{P_1/(1-P_1)}{P_0/(1-P_0)}~~&\text{if } P_1>P_0,\\
1=\frac{P_1^\ast/(1-P_1^\ast)}{P_0^\ast/(1-P_0^\ast)} = \frac{P_1/(1-P_1)}{P_0/(1-P_0)}~~&\text{if } P_1=P_0,\text{ and}\\
1>\frac{P_1^\ast/(1-P_1^\ast)}{P_0^\ast/(1-P_0^\ast)} > \frac{P_1/(1-P_1)}{P_0/(1-P_0)}~~&\text{if } P_1<P_0
\end{align*}
\end{theorem}
\begin{proof}
Suppose $(P_0,P_1)\in (0,1)\times(0,1)$. If and only if
\begin{align}
\frac{P_1^\ast/(1-P_1^\ast)}{P_0^\ast/(1-P_0^\ast)} <\frac{P_1/(1-P_1)}{P_0/(1-P_0)}, \label{e:inequality}
\end{align}
then
\begin{align}
\frac{P_1s+1}{s+1-P_1s}\frac{s+1-P_0s}{P_0s+1} &< \frac{P_1}{1-P_1}\frac{1-P_0}{P_0},\nonumber
\\
\frac{P_1s+1}{s+1-P_1s}\frac{1-P_1}{P_1} &< \frac{P_0s+1}{s+1-P_0s}\frac{1-P_0}{P_0}.\nonumber
\end{align}
Now, since
\begin{align*}
\frac{\partial}{\partial P}\bigg\{\frac{Ps+1}{s+1-Ps}\frac{1-P}{P}\bigg\} = \frac{(-2P^2+2P-1)S-1}{P^2(1-(P-1)S)^2} < 0
\end{align*}
over the interval $(0,1)$ for $P$, it follows that inequality \eqref{e:inequality} holds if $P_1>P_0$. Also, if $P_1>P_0$, then, since $\partial/(\partial P)\{ (Ps+1)/(s+1-Ps)\}>0$ if $P\in(0,1)$, we have $$1<\frac{P_1^\ast/(1-P_1^\ast)}{P_0^\ast/(1-P_0^\ast)}.$$
Similar arguments establish the assertion for the case where $P_1<P_0$. It is easily verified that if $P_1=P_0$, then
\begin{align*}
\frac{P_1^\ast/(1-P_1^\ast)}{P_0^\ast/(1-P_0^\ast)}&=\frac{P_1s+1}{s+1-P_1s}\frac{s+1-P_0s}{P_0s+1}=1 \\&=\frac{P_1}{1-P_1}\frac{1-P_0}{P_0}=\frac{P_1/(1-P_1)}{P_0/(1-P_0)},
\end{align*}
as desired.
\end{proof}
\section{III}\label{s:app_Rcode}
GP and IPWM were applied to every dataset \texttt{data} in R using the function \texttt{mecor::ipwm} and the following code:\\\par
\noindent\texttt{%
\# GP: \\
formulasGP <- list(\\
\indent Y{\textasciitilde}Z+B+L1+L2+L3+L4+L5+L6+L7+L8+L9+L10,\\
\indent B{\textasciitilde}Z+L1+L2+L3+L4+L5+L6+L7+L8+L9+L10,\\
\indent Z{\textasciitilde}L1+L2+L3+L4+L5+L6+L7+L8+L9+L10\\
)\\
mecor::ipwm(\\ \indent
formulas=formulasGP, data=data, outcome\_true=``Y'', \\
\indent outcome\_mis=``Z'', exposure\_true=``B'', exposure\_mis=NULL, sp=1e6\\
)\\
\ \\
\# IPWM: \\
formulasIPWM <- list(\\ \indent Y{\textasciitilde}A+Z+B+L1+L2+L3+L4+L5+L6+L7+L8+L9+L10,\\
\indent A{\textasciitilde}Z+B+L1+L2+L3+L4+L5+L6+L7+L8+L9+L10,\\
\indent Z{\textasciitilde}B+L1+L2+L3+L4+L5+L6+L7+L8+L9+L10,\\
\indent B{\textasciitilde}L1+L2+L3+L4+L5+L6+L7+L8+L9+L10\\
)\\
mecor::ipwm(\\ \indent
formulas=formulasIPWM, data=data, outcome\_true=``Y'', \\
\indent outcome\_mis=``Z'', exposure\_true=``A'', exposure\_mis=``B'', sp=1e6\\
)
}\\\par\noindent

\label{lastpage}
\end{document}



\pagerange{\pageref{firstpage}--\pageref{lastpage}} 


\label{firstpage}



\maketitle

\begin{sidewaystable}
\centering
\caption{Log-likelihood contributions for all possible types of observations under internal validation sampling.\break }
\begin{tabular*}{\textwidth}{
@{\extracolsep{\fill}}cccccccccl@{\extracolsep{\fill}}
}
\toprule
Type&$R_Y$&$R_A$&$Z$&$B$&$Y$&$A$&$L$&Count&Log-likelihood contribution
\\\midrule
1&0&0&0&0& & &0&$m_{1}$&$\log(1-\varepsilon^\ast_{00})+\log(1-\delta^\ast_{0})$\\
2&0&0&1&0& & &0&$m_{2}$&$\log(\varepsilon^\ast_{00})+\log(1-\delta^\ast_{0})$\\
3&0&0&0&1& & &0&$m_{3}$&$\log(1-\varepsilon^\ast_{10})+\log(\delta^\ast_{0})$\\
4&0&0&1&1& & &0&$m_{4}$&$\log(\varepsilon^\ast_{10})+\log(\delta^\ast_{0})$\\
5&0&0&0&0& & &1&$m_{5}$&$\log(1-\varepsilon^\ast_{01})+\log(1-\delta^\ast_{1})$\\
6&0&0&1&0& & &1&$m_{6}$&$\log(\varepsilon^\ast_{01})+\log(1-\delta^\ast_{1})$\\
7&0&0&0&1& & &1&$m_{7}$&$\log(1-\varepsilon^\ast_{11})+\log(\delta^\ast_{1})$\\
8&0&0&1&1& & &1&$m_{8}$&$\log(\varepsilon^\ast_{11})+\log(\delta^\ast_{1})$\\
9&0&1& & & & & &$m_{9}+...+m_{12}=0$&0\\
10&1&0& & & & & &$m_{13}+...+m_{16}=0$&0\\
11&1&1&0&0&0&0&0&$m_{17}$&$\log(1-\pi^\ast_{0000})+\log(1-\lambda^\ast_{000})+\log(1-\varepsilon^\ast_{00})+\log(1-\delta^\ast_{0})$\\
12&1&1&1&0&0&0&0&$m_{18}$&$\log(1-\pi^\ast_{0100})+\log(1-\lambda^\ast_{100})+\log(\varepsilon^\ast_{00})+\log(1-\delta^\ast_{0})$\\
13&1&1&0&1&0&0&0&$m_{19}$&$\log(1-\pi^\ast_{0010})+\log(1-\lambda^\ast_{010})+\log(1-\varepsilon^\ast_{10})+\log(\delta^\ast_{0})$\\
14&1&1&1&1&0&0&0&$m_{20}$&$\log(1-\pi^\ast_{0110})+\log(1-\lambda^\ast_{110})+\log(\varepsilon^\ast_{10})+\log(\delta^\ast_{0})$\\
15&1&1&0&0&1&0&0&$m_{21}$&$\log(\pi^\ast_{0000})+\log(1-\lambda^\ast_{000})+\log(1-\varepsilon^\ast_{00})+\log(1-\delta^\ast_{0})$\\
16&1&1&1&0&1&0&0&$m_{22}$&$\log(\pi^\ast_{0100})+\log(1-\lambda^\ast_{100})+\log(\varepsilon^\ast_{00})+\log(1-\delta^\ast_{0})$\\
17&1&1&0&1&1&0&0&$m_{23}$&$\log(\pi^\ast_{0010})+\log(1-\lambda^\ast_{010})+\log(1-\varepsilon^\ast_{10})+\log(\delta^\ast_{0})$\\
18&1&1&1&1&1&0&0&$m_{24}$&$\log(\pi^\ast_{0110})+\log(1-\lambda^\ast_{110})+\log(\varepsilon^\ast_{10})+\log(\delta^\ast_{0})$\\
19&1&1&0&0&0&1&0&$m_{25}$&$\log(1-\pi^\ast_{1000})+\log(\lambda^\ast_{000})+\log(1-\varepsilon^\ast_{00})+\log(1-\delta^\ast_{0})$\\
20&1&1&1&0&0&1&0&$m_{26}$&$\log(1-\pi^\ast_{1100})+\log(\lambda^\ast_{100})+\log(\varepsilon^\ast_{00})+\log(1-\delta^\ast_{0})$\\
21&1&1&0&1&0&1&0&$m_{27}$&$\log(1-\pi^\ast_{1010})+\log(\lambda^\ast_{010})+\log(1-\varepsilon^\ast_{10})+\log(\delta^\ast_{0})$\\
22&1&1&1&1&0&1&0&$m_{28}$&$\log(1-\pi^\ast_{1110})+\log(\lambda^\ast_{110})+\log(\varepsilon^\ast_{10})+\log(\delta^\ast_{0})$\\
23&1&1&0&0&1&1&0&$m_{29}$&$\log(\pi^\ast_{1000})+\log(\lambda^\ast_{000})+\log(1-\varepsilon^\ast_{00})+\log(1-\delta^\ast_{0})$\\
24&1&1&1&0&1&1&0&$m_{30}$&$\log(\pi^\ast_{1100})+\log(\lambda^\ast_{100})+\log(\varepsilon^\ast_{00})+\log(1-\delta^\ast_{0})$\\
25&1&1&0&1&1&1&0&$m_{31}$&$\log(\pi^\ast_{1010})+\log(\lambda^\ast_{010})+\log(1-\varepsilon^\ast_{10})+\log(\delta^\ast_{0})$\\
26&1&1&1&1&1&1&0&$m_{32}$&$\log(\pi^\ast_{1110})+\log(\lambda^\ast_{110})+\log(\varepsilon^\ast_{10})+\log(\delta^\ast_{0})$\\
27&1&1&0&0&0&0&1&$m_{33}$&$\log(1-\pi^\ast_{0001})+\log(1-\lambda^\ast_{001})+\log(1-\varepsilon^\ast_{01})+\log(1-\delta^\ast_{1})$\\
28&1&1&1&0&0&0&1&$m_{34}$&$\log(1-\pi^\ast_{0101})+\log(1-\lambda^\ast_{101})+\log(\varepsilon^\ast_{01})+\log(1-\delta^\ast_{1})$\\
29&1&1&0&1&0&0&1&$m_{35}$&$\log(1-\pi^\ast_{0011})+\log(1-\lambda^\ast_{011})+\log(1-\varepsilon^\ast_{11})+\log(\delta^\ast_{1})$\\
30&1&1&1&1&0&0&1&$m_{36}$&$\log(1-\pi^\ast_{0111})+\log(1-\lambda^\ast_{111})+\log(\varepsilon^\ast_{11})+\log(\delta^\ast_{1})$\\
\end{tabular*}
\end{sidewaystable}

\begin{sidewaystable}
\centering
\contcaption{(continued.)\\~}
\begin{tabular*}{\textwidth}{
@{\extracolsep{\fill}}cccccccccl@{\extracolsep{\fill}}
}
\toprule
Type&$R_Y$&$R_A$&$Z$&$B$&$Y$&$A$&$L$&Count&Log-likelihood contribution
\\\midrule
31&1&1&0&0&1&0&1&$m_{37}$&$\log(\pi^\ast_{0001})+\log(1-\lambda^\ast_{001})+\log(1-\varepsilon^\ast_{01})+\log(1-\delta^\ast_{1})$\\
32&1&1&1&0&1&0&1&$m_{38}$&$\log(\pi^\ast_{0101})+\log(1-\lambda^\ast_{101})+\log(\varepsilon^\ast_{01})+\log(1-\delta^\ast_{1})$\\
33&1&1&0&1&1&0&1&$m_{39}$&$\log(\pi^\ast_{0011})+\log(1-\lambda^\ast_{011})+\log(1-\varepsilon^\ast_{11})+\log(\delta^\ast_{1})$\\
34&1&1&1&1&1&0&1&$m_{40}$&$\log(\pi^\ast_{0111})+\log(1-\lambda^\ast_{111})+\log(\varepsilon^\ast_{11})+\log(\delta^\ast_{1})$\\
35&1&1&0&0&0&1&1&$m_{41}$&$\log(1-\pi^\ast_{1001})+\log(\lambda^\ast_{001})+\log(1-\varepsilon^\ast_{01})+\log(1-\delta^\ast_{1})$\\
36&1&1&1&0&0&1&1&$m_{42}$&$\log(1-\pi^\ast_{1101})+\log(\lambda^\ast_{101})+\log(\varepsilon^\ast_{01})+\log(1-\delta^\ast_{1})$\\
37&1&1&0&1&0&1&1&$m_{43}$&$\log(1-\pi^\ast_{1011})+\log(\lambda^\ast_{011})+\log(1-\varepsilon^\ast_{11})+\log(\delta^\ast_{1})$\\
38&1&1&1&1&0&1&1&$m_{44}$&$\log(1-\pi^\ast_{1111})+\log(\lambda^\ast_{111})+\log(\varepsilon^\ast_{11})+\log(\delta^\ast_{1})$\\
39&1&1&0&0&1&1&1&$m_{45}$&$\log(\pi^\ast_{1001})+\log(\lambda^\ast_{001})+\log(1-\varepsilon^\ast_{01})+\log(1-\delta^\ast_{1})$\\
40&1&1&1&0&1&1&1&$m_{46}$&$\log(\pi^\ast_{1101})+\log(\lambda^\ast_{101})+\log(\varepsilon^\ast_{01})+\log(1-\delta^\ast_{1})$\\
41&1&1&0&1&1&1&1&$m_{47}$&$\log(\pi^\ast_{1011})+\log(\lambda^\ast_{011})+\log(1-\varepsilon^\ast_{11})+\log(\delta^\ast_{1})$\\
42&1&1&1&1&1&1&1&$m_{48}$&$\log(\pi^\ast_{1111})+\log(\lambda^\ast_{111})+\log(\varepsilon^\ast_{11})+\log(\delta^\ast_{1})$\\
\bottomrule
\end{tabular*}
\end{sidewaystable}

\begin{sidewaystable}
\centering
\caption{Closed form expressions for the maximum likelihood estimators (MLE) of the parameters of the likelihood parameterised in terms of predictive values for the hypothetical study setting.\break }
\begin{tabular*}{\textwidth}{
@{\extracolsep{\fill}}l@{\extracolsep{1cm}}l@{\extracolsep{\fill}}
}
\toprule
Parameter&MLE\\
\midrule
$\delta^\ast_0$ & $\hat{\delta}^\ast_0$ = $(m_{3}+m_{4}+m_{19}+m_{20}+m_{23}+m_{24}+m_{27}+m_{28}+m_{31}+m_{32})/(\{\sum_{j=1}^4m_j\}+\{\sum_{j=17}^{32}m_j\})$ \\
$\delta^\ast_1$ & $\hat{\delta}^\ast_1$ = $(m_{7}+m_{8}+m_{35}+m_{36}+m_{39}+m_{40}+m_{43}+m_{44}+m_{47}+m_{48})/(\{\sum_{j=5}^8m_j\}+\{\sum_{j=33}^{48}m_j\})$ \\
$\varepsilon^\ast_{00}$ & $\hat{\varepsilon}^\ast_{00}=(m_2+m_{18}+m_{22}+m_{26}+m_{30})/(m_1+m_2+m_{17}+m_{18}+m_{21}+m_{22}+m_{25}+m_{26}+m_{29}+m_{30})$\\
$\varepsilon^\ast_{10}$ & $\hat{\varepsilon}^\ast_{10}=(m_4+m_{20}+m_{24}+m_{28}+m_{32})/(m_3+m_4+m_{19}+m_{20}+m_{23}+m_{24}+m_{27}+m_{28}+m_{31}+m_{32})$\\
$\varepsilon^\ast_{01}$ & $\hat{\varepsilon}^\ast_{01}=(m_6+m_{34}+m_{38}+m_{42}+m_{46})/(m_5+m_6+m_{33}+m_{34}+m_{37}+m_{38}+m_{41}+m_{42}+m_{45}+m_{46})$\\
$\varepsilon^\ast_{11}$ & $\hat{\varepsilon}^\ast_{11}=(m_8+m_{36}+m_{40}+m_{44}+m_{48})/(m_7+m_8+m_{35}+m_{36}+m_{39}+m_{40}+m_{43}+m_{44}+m_{47}+m_{48})$\\
$\lambda^\ast_{000}$ & $\hat{\lambda}^\ast_{000}=(m_{25}+m_{29})/(m_{17}+m_{21}+m_{25}+m_{29})$\\
$\lambda^\ast_{100}$ & $\hat{\lambda}^\ast_{100}=(m_{26}+m_{30})/(m_{18}+m_{22}+m_{26}+m_{30})$\\
$\lambda^\ast_{010}$ & $\hat{\lambda}^\ast_{010}=(m_{27}+m_{31})/(m_{19}+m_{23}+m_{27}+m_{31})$\\
$\lambda^\ast_{110}$ & $\hat{\lambda}^\ast_{110}=(m_{28}+m_{32})/(m_{20}+m_{24}+m_{28}+m_{32})$\\
$\lambda^\ast_{001}$ & $\hat{\lambda}^\ast_{001}=(m_{41}+m_{45})/(m_{33}+m_{37}+m_{41}+m_{45})$\\
$\lambda^\ast_{101}$ & $\hat{\lambda}^\ast_{101}=(m_{42}+m_{46})/(m_{34}+m_{38}+m_{42}+m_{46})$\\
$\lambda^\ast_{011}$ & $\hat{\lambda}^\ast_{011}=(m_{43}+m_{47})/(m_{35}+m_{39}+m_{43}+m_{47})$\\
$\lambda^\ast_{111}$ & $\hat{\lambda}^\ast_{111}=(m_{44}+m_{48})/(m_{36}+m_{40}+m_{44}+m_{48})$\\
$\pi^\ast_{0000}$ & $\hat{\pi}^\ast_{0000}=m_{21}/(m_{17}+m_{21})$\\
$\pi^\ast_{1000}$ & $\hat{\pi}^\ast_{1000}=m_{29}/(m_{25}+m_{29})$\\
$\pi^\ast_{0100}$ & $\hat{\pi}^\ast_{0100}=m_{22}/(m_{18}+m_{22})$\\
$\pi^\ast_{1100}$ & $\hat{\pi}^\ast_{1100}=m_{30}/(m_{26}+m_{30})$\\
$\pi^\ast_{0010}$ & $\hat{\pi}^\ast_{0010}=m_{23}/(m_{19}+m_{23})$\\
$\pi^\ast_{1010}$ & $\hat{\pi}^\ast_{1010}=m_{31}/(m_{27}+m_{31})$\\
$\pi^\ast_{0110}$ & $\hat{\pi}^\ast_{0110}=m_{24}/(m_{20}+m_{24})$\\
$\pi^\ast_{1110}$ & $\hat{\pi}^\ast_{1110}=m_{32}/(m_{28}+m_{32})$\\
$\pi^\ast_{0001}$ & $\hat{\pi}^\ast_{0001}=m_{37}/(m_{33}+m_{37})$\\
$\pi^\ast_{1001}$ & $\hat{\pi}^\ast_{1001}=m_{45}/(m_{41}+m_{45})$\\
$\pi^\ast_{0101}$ & $\hat{\pi}^\ast_{0101}=m_{38}/(m_{34}+m_{38})$\\
$\pi^\ast_{1101}$ & $\hat{\pi}^\ast_{1101}=m_{46}/(m_{42}+m_{46})$\\
$\pi^\ast_{0011}$ & $\hat{\pi}^\ast_{0011}=m_{39}/(m_{35}+m_{39})$\\
$\pi^\ast_{1011}$ & $\hat{\pi}^\ast_{1011}=m_{47}/(m_{43}+m_{47})$\\
$\pi^\ast_{0111}$ & $\hat{\pi}^\ast_{0111}=m_{40}/(m_{36}+m_{40})$\\
$\pi^\ast_{1111}$ & $\hat{\pi}^\ast_{1111}=m_{48}/(m_{44}+m_{48})$\\
\bottomrule
\end{tabular*}
\end{sidewaystable}

\begin{table}
\centering
\caption{Results for simulation studies 1-36 on the performance of different causal estimators in various scenarios of confounding and misclassification in exposure and outcome. Abbreviations: PS, propensity score method ignoring misclassification; CCA, complete case analysis; GP, Gravel and Platt estimator ignoring exposure misclassification; IPWM, inverse probability weighting method for confounding and joint exposure and outcome misclassification; BSE, estimated standard error for the bias due to Monte Carlo error; SE, empirical standard error; SSE, sample standard error; CP, empirical coverage probability. In all scenarios, the true marginal log OR (estimand) was $-0.4$.\break
}
\begin{tabular*}{\textwidth}{
@{\extracolsep{\fill}}lcccccclcccccc@{\extracolsep{\fill}}
}
\toprule
&\multicolumn{6}{c}{Crude}&&\multicolumn{6}{c}{PS}\\\cline{2-7}\cline{9-14}
Scenario&Bias&BSE&MSE&SE&SSE&CP&&Bias&BSE&MSE&SE&SSE&CP
\\\midrule
1 & $\phantom{-}0.394$ & $0.004$ & $0.275$ & $0.119$ & $0.118$ & $0.122$ && $\phantom{-}0.392$ & $0.005$ & $0.321$ & $0.168$ & $0.169$ & $0.382$ \\
2 & $\phantom{-}0.382$ & $0.006$ & $0.329$ & $0.183$ & $0.184$ & $0.492$ && $\phantom{-}0.379$ & $0.008$ & $0.407$ & $0.264$ & $0.258$ & $0.738$ \\
3 & $\phantom{-}0.394$ & $0.004$ & $0.272$ & $0.117$ & $0.118$ & $0.116$ && $\phantom{-}0.389$ & $0.006$ & $0.327$ & $0.175$ & $0.169$ & $0.402$ \\
4 & $\phantom{-}0.401$ & $0.004$ & $0.278$ & $0.117$ & $0.118$ & $0.102$ && $\phantom{-}0.389$ & $0.006$ & $0.327$ & $0.176$ & $0.168$ & $0.392$ \\
5 & $\phantom{-}0.401$ & $0.003$ & $0.252$ & $0.090$ & $0.088$ & $0.007$ && $\phantom{-}0.402$ & $0.003$ & $0.252$ & $0.090$ & $0.088$ & $0.010$ \\
6 & $\phantom{-}0.407$ & $0.004$ & $0.298$ & $0.132$ & $0.134$ & $0.133$ && $\phantom{-}0.407$ & $0.004$ & $0.297$ & $0.131$ & $0.135$ & $0.136$ \\
7 & $\phantom{-}0.396$ & $0.003$ & $0.243$ & $0.086$ & $0.088$ & $0.009$ && $\phantom{-}0.396$ & $0.003$ & $0.243$ & $0.086$ & $0.088$ & $0.009$ \\
8 & $\phantom{-}0.395$ & $0.003$ & $0.242$ & $0.086$ & $0.088$ & $0.005$ && $\phantom{-}0.395$ & $0.003$ & $0.242$ & $0.086$ & $0.088$ & $0.004$ \\
9 & $\phantom{-}0.398$ & $0.003$ & $0.247$ & $0.089$ & $0.088$ & $0.005$ && $\phantom{-}0.398$ & $0.003$ & $0.247$ & $0.089$ & $0.088$ & $0.005$ \\
10 & $\phantom{-}0.401$ & $0.003$ & $0.242$ & $0.081$ & $0.083$ & $0.004$ && $\phantom{-}0.399$ & $0.004$ & $0.276$ & $0.117$ & $0.120$ & $0.080$ \\
11 & $\phantom{-}0.392$ & $0.004$ & $0.281$ & $0.127$ & $0.127$ & $0.189$ && $\phantom{-}0.391$ & $0.006$ & $0.330$ & $0.177$ & $0.181$ & $0.436$ \\
12 & $\phantom{-}0.400$ & $0.003$ & $0.243$ & $0.083$ & $0.083$ & $0.005$ && $\phantom{-}0.391$ & $0.004$ & $0.273$ & $0.119$ & $0.119$ & $0.104$ \\
13 & $\phantom{-}0.394$ & $0.003$ & $0.237$ & $0.081$ & $0.083$ & $0.007$ && $\phantom{-}0.392$ & $0.004$ & $0.276$ & $0.122$ & $0.120$ & $0.106$ \\
14 & $\phantom{-}0.398$ & $0.002$ & $0.219$ & $0.061$ & $0.062$ & $0.000$ && $\phantom{-}0.398$ & $0.002$ & $0.219$ & $0.061$ & $0.062$ & $0.000$ \\
15 & $\phantom{-}0.404$ & $0.003$ & $0.257$ & $0.094$ & $0.094$ & $0.010$ && $\phantom{-}0.404$ & $0.003$ & $0.257$ & $0.094$ & $0.095$ & $0.011$ \\
16 & $\phantom{-}0.399$ & $0.002$ & $0.222$ & $0.062$ & $0.062$ & $0.000$ && $\phantom{-}0.399$ & $0.002$ & $0.222$ & $0.062$ & $0.062$ & $0.000$ \\
17 & $\phantom{-}0.401$ & $0.002$ & $0.225$ & $0.064$ & $0.062$ & $0.000$ && $\phantom{-}0.401$ & $0.002$ & $0.225$ & $0.064$ & $0.062$ & $0.000$ \\
18 & $\phantom{-}0.400$ & $0.002$ & $0.224$ & $0.064$ & $0.062$ & $0.000$ && $\phantom{-}0.400$ & $0.002$ & $0.224$ & $0.064$ & $0.062$ & $0.000$ \\
19 & $\phantom{-}0.397$ & $0.004$ & $0.274$ & $0.117$ & $0.118$ & $0.100$ && $\phantom{-}0.396$ & $0.005$ & $0.324$ & $0.167$ & $0.170$ & $0.387$ \\
20 & $\phantom{-}0.391$ & $0.006$ & $0.332$ & $0.179$ & $0.183$ & $0.466$ && $\phantom{-}0.362$ & $0.008$ & $0.392$ & $0.261$ & $0.253$ & $0.732$ \\
21 & $\phantom{-}0.401$ & $0.004$ & $0.278$ & $0.118$ & $0.118$ & $0.109$ && $\phantom{-}0.391$ & $0.005$ & $0.325$ & $0.173$ & $0.169$ & $0.394$ \\
22 & $\phantom{-}0.404$ & $0.004$ & $0.274$ & $0.111$ & $0.117$ & $0.080$ && $\phantom{-}0.396$ & $0.005$ & $0.326$ & $0.169$ & $0.167$ & $0.367$ \\
23 & $\phantom{-}0.400$ & $0.003$ & $0.247$ & $0.087$ & $0.088$ & $0.008$ && $\phantom{-}0.400$ & $0.003$ & $0.247$ & $0.087$ & $0.088$ & $0.006$ \\
24 & $\phantom{-}0.397$ & $0.004$ & $0.292$ & $0.135$ & $0.134$ & $0.161$ && $\phantom{-}0.397$ & $0.004$ & $0.292$ & $0.135$ & $0.135$ & $0.161$ \\
25 & $\phantom{-}0.401$ & $0.003$ & $0.248$ & $0.087$ & $0.088$ & $0.006$ && $\phantom{-}0.400$ & $0.003$ & $0.248$ & $0.087$ & $0.088$ & $0.006$ \\
26 & $\phantom{-}0.403$ & $0.003$ & $0.249$ & $0.087$ & $0.088$ & $0.003$ && $\phantom{-}0.403$ & $0.003$ & $0.249$ & $0.087$ & $0.088$ & $0.004$ \\
27 & $\phantom{-}0.400$ & $0.003$ & $0.247$ & $0.087$ & $0.088$ & $0.004$ && $\phantom{-}0.400$ & $0.003$ & $0.248$ & $0.088$ & $0.088$ & $0.003$ \\
28 & $\phantom{-}0.396$ & $0.003$ & $0.243$ & $0.085$ & $0.083$ & $0.004$ && $\phantom{-}0.395$ & $0.004$ & $0.279$ & $0.123$ & $0.119$ & $0.101$ \\
29 & $\phantom{-}0.396$ & $0.004$ & $0.284$ & $0.128$ & $0.127$ & $0.176$ && $\phantom{-}0.388$ & $0.006$ & $0.337$ & $0.187$ & $0.182$ & $0.455$ \\
30 & $\phantom{-}0.398$ & $0.003$ & $0.239$ & $0.081$ & $0.083$ & $0.007$ && $\phantom{-}0.398$ & $0.004$ & $0.279$ & $0.120$ & $0.120$ & $0.096$ \\
31 & $\phantom{-}0.399$ & $0.003$ & $0.243$ & $0.083$ & $0.083$ & $0.004$ && $\phantom{-}0.395$ & $0.004$ & $0.276$ & $0.120$ & $0.119$ & $0.102$ \\
32 & $\phantom{-}0.404$ & $0.002$ & $0.224$ & $0.061$ & $0.062$ & $0.000$ && $\phantom{-}0.404$ & $0.002$ & $0.224$ & $0.061$ & $0.062$ & $0.000$ \\
33 & $\phantom{-}0.398$ & $0.003$ & $0.251$ & $0.092$ & $0.094$ & $0.011$ && $\phantom{-}0.398$ & $0.003$ & $0.251$ & $0.092$ & $0.095$ & $0.012$ \\
34 & $\phantom{-}0.404$ & $0.002$ & $0.226$ & $0.063$ & $0.062$ & $0.000$ && $\phantom{-}0.404$ & $0.002$ & $0.226$ & $0.063$ & $0.062$ & $0.000$ \\
35 & $\phantom{-}0.399$ & $0.002$ & $0.221$ & $0.061$ & $0.062$ & $0.000$ && $\phantom{-}0.399$ & $0.002$ & $0.221$ & $0.061$ & $0.062$ & $0.000$ \\
36 & $\phantom{-}0.401$ & $0.002$ & $0.220$ & $0.059$ & $0.062$ & $0.000$ && $\phantom{-}0.401$ & $0.002$ & $0.220$ & $0.059$ & $0.062$ & $0.000$ \\
\midrule
\end{tabular*}
\end{table}

\begin{table}
\centering
\contcaption{(Continued.)\break
}
\begin{tabular*}{\textwidth}{
@{\extracolsep{\fill}}lcccccclcccccc@{\extracolsep{\fill}}
}
\toprule
&\multicolumn{6}{c}{CCA}&&\multicolumn{6}{c}{GP}\\\cline{2-7}\cline{9-14}
Scenario&Bias&BSE&MSE&SE&SSE&CP&&Bias&BSE&MSE&SE&SSE&CP
\\\midrule
1 & $-0.078$ & $0.015$ & $0.476$ & $0.469$ & $0.491$ & $0.899$ && $-0.036$ & $0.011$ & $0.360$ & $0.359$ & $0.428$ & $0.958$ \\
2 & $-0.117$ & $0.019$ & $0.615$ & $0.601$ & $0.900$ & $0.887$ && $-0.097$ & $0.016$ & $0.515$ & $0.505$ & $0.861$ & $0.938$ \\
3 & $-0.020$ & $0.010$ & $0.301$ & $0.301$ & $0.300$ & $0.919$ && $-0.019$ & $0.007$ & $0.234$ & $0.233$ & $0.240$ & $0.939$ \\
4 & $-0.093$ & $0.020$ & $0.640$ & $0.631$ & $1.158$ & $0.899$ && $-0.045$ & $0.016$ & $0.503$ & $0.501$ & $1.087$ & $0.944$ \\
5 & $-0.145$ & $0.009$ & $0.307$ & $0.286$ & $0.286$ & $0.903$ && $\phantom{-}0.269$ & $0.008$ & $0.316$ & $0.244$ & $0.244$ & $0.799$ \\
6 & $-0.109$ & $0.011$ & $0.357$ & $0.345$ & $0.362$ & $0.930$ && $\phantom{-}0.280$ & $0.010$ & $0.392$ & $0.314$ & $0.339$ & $0.862$ \\
7 & $-0.213$ & $0.007$ & $0.282$ & $0.237$ & $0.250$ & $0.865$ && $\phantom{-}0.134$ & $0.008$ & $0.259$ & $0.241$ & $0.252$ & $0.926$ \\
8 & $-0.209$ & $0.006$ & $0.231$ & $0.187$ & $0.186$ & $0.775$ && $\phantom{-}0.259$ & $0.004$ & $0.207$ & $0.140$ & $0.144$ & $0.570$ \\
9 & $-0.175$ & $0.012$ & $0.422$ & $0.392$ & $0.411$ & $0.902$ && $\phantom{-}0.263$ & $0.010$ & $0.394$ & $0.325$ & $0.339$ & $0.883$ \\
10 & $-0.011$ & $0.010$ & $0.302$ & $0.302$ & $0.315$ & $0.932$ && $-0.016$ & $0.008$ & $0.242$ & $0.242$ & $0.257$ & $0.960$ \\
11 & $-0.038$ & $0.013$ & $0.406$ & $0.404$ & $0.398$ & $0.909$ && $-0.024$ & $0.010$ & $0.328$ & $0.328$ & $0.353$ & $0.956$ \\
12 & $\phantom{-}0.004$ & $0.007$ & $0.210$ & $0.210$ & $0.208$ & $0.939$ && $-0.013$ & $0.005$ & $0.167$ & $0.167$ & $0.165$ & $0.943$ \\
13 & $-0.050$ & $0.014$ & $0.435$ & $0.432$ & $0.441$ & $0.905$ && $-0.022$ & $0.011$ & $0.341$ & $0.341$ & $0.371$ & $0.944$ \\
14 & $-0.128$ & $0.006$ & $0.211$ & $0.194$ & $0.199$ & $0.890$ && $\phantom{-}0.271$ & $0.005$ & $0.236$ & $0.163$ & $0.168$ & $0.633$ \\
15 & $-0.097$ & $0.007$ & $0.246$ & $0.237$ & $0.245$ & $0.926$ && $\phantom{-}0.269$ & $0.007$ & $0.293$ & $0.221$ & $0.225$ & $0.772$ \\
16 & $-0.232$ & $0.005$ & $0.221$ & $0.168$ & $0.173$ & $0.736$ && $\phantom{-}0.118$ & $0.005$ & $0.185$ & $0.171$ & $0.173$ & $0.904$ \\
17 & $-0.197$ & $0.004$ & $0.169$ & $0.130$ & $0.130$ & $0.646$ && $\phantom{-}0.263$ & $0.003$ & $0.168$ & $0.098$ & $0.101$ & $0.261$ \\
18 & $-0.173$ & $0.008$ & $0.296$ & $0.266$ & $0.270$ & $0.883$ && $\phantom{-}0.257$ & $0.007$ & $0.290$ & $0.224$ & $0.229$ & $0.795$ \\
19 & $\phantom{-}0.011$ & $0.007$ & $0.237$ & $0.237$ & $0.244$ & $0.957$ && $-0.002$ & $0.007$ & $0.223$ & $0.223$ & $0.221$ & $0.939$ \\
20 & $\phantom{-}0.001$ & $0.009$ & $0.288$ & $0.288$ & $0.276$ & $0.918$ && $-0.019$ & $0.010$ & $0.305$ & $0.304$ & $0.304$ & $0.938$ \\
21 & $\phantom{-}0.058$ & $0.007$ & $0.219$ & $0.216$ & $0.223$ & $0.953$ && $-0.007$ & $0.006$ & $0.194$ & $0.194$ & $0.197$ & $0.949$ \\
22 & $-0.015$ & $0.011$ & $0.350$ & $0.350$ & $0.345$ & $0.934$ && $-0.023$ & $0.009$ & $0.277$ & $0.277$ & $0.287$ & $0.950$ \\
23 & $-0.092$ & $0.005$ & $0.154$ & $0.146$ & $0.148$ & $0.889$ && $\phantom{-}0.263$ & $0.004$ & $0.205$ & $0.136$ & $0.136$ & $0.505$ \\
24 & $-0.060$ & $0.005$ & $0.174$ & $0.170$ & $0.170$ & $0.929$ && $\phantom{-}0.263$ & $0.006$ & $0.247$ & $0.177$ & $0.183$ & $0.712$ \\
25 & $-0.171$ & $0.004$ & $0.161$ & $0.132$ & $0.131$ & $0.741$ && $\phantom{-}0.139$ & $0.004$ & $0.155$ & $0.136$ & $0.138$ & $0.820$ \\
26 & $-0.121$ & $0.004$ & $0.148$ & $0.134$ & $0.135$ & $0.842$ && $\phantom{-}0.263$ & $0.004$ & $0.184$ & $0.115$ & $0.118$ & $0.388$ \\
27 & $-0.113$ & $0.007$ & $0.219$ & $0.206$ & $0.207$ & $0.904$ && $\phantom{-}0.264$ & $0.006$ & $0.248$ & $0.178$ & $0.185$ & $0.702$ \\
28 & $\phantom{-}0.017$ & $0.005$ & $0.169$ & $0.169$ & $0.170$ & $0.953$ && $\phantom{-}0.003$ & $0.005$ & $0.147$ & $0.147$ & $0.152$ & $0.946$ \\
29 & $\phantom{-}0.007$ & $0.006$ & $0.200$ & $0.200$ & $0.193$ & $0.947$ && $-0.014$ & $0.006$ & $0.196$ & $0.196$ & $0.203$ & $0.952$ \\
30 & $\phantom{-}0.058$ & $0.005$ & $0.161$ & $0.157$ & $0.154$ & $0.928$ && $-0.003$ & $0.004$ & $0.138$ & $0.138$ & $0.136$ & $0.943$ \\
31 & $\phantom{-}0.018$ & $0.007$ & $0.236$ & $0.236$ & $0.237$ & $0.946$ && $-0.003$ & $0.006$ & $0.193$ & $0.192$ & $0.194$ & $0.940$ \\
32 & $-0.092$ & $0.003$ & $0.107$ & $0.099$ & $0.105$ & $0.864$ && $\phantom{-}0.265$ & $0.003$ & $0.161$ & $0.091$ & $0.095$ & $0.191$ \\
33 & $-0.051$ & $0.004$ & $0.118$ & $0.115$ & $0.119$ & $0.933$ && $\phantom{-}0.264$ & $0.004$ & $0.191$ & $0.121$ & $0.124$ & $0.421$ \\
34 & $-0.166$ & $0.003$ & $0.121$ & $0.094$ & $0.092$ & $0.559$ && $\phantom{-}0.138$ & $0.003$ & $0.115$ & $0.096$ & $0.096$ & $0.710$ \\
35 & $-0.116$ & $0.003$ & $0.108$ & $0.095$ & $0.094$ & $0.762$ && $\phantom{-}0.264$ & $0.003$ & $0.150$ & $0.080$ & $0.082$ & $0.110$ \\
36 & $-0.115$ & $0.005$ & $0.163$ & $0.149$ & $0.144$ & $0.859$ && $\phantom{-}0.266$ & $0.004$ & $0.202$ & $0.131$ & $0.128$ & $0.455$ \\
\midrule
\end{tabular*}
\end{table}

\begin{table}
\centering
\contcaption{(Continued.)\break
}
\begin{tabular*}{\textwidth}{
@{\extracolsep{\fill}}lcccccc@{\extracolsep{8.2mm}}lcccccc@{\extracolsep{\fill}}
}
\toprule
&\multicolumn{6}{c}{IPWM}\\\cline{2-7}
Scenario&Bias&BSE&MSE&SE&SSE&CP&&&&&&&
\\\midrule
1 & $-0.036$ & $0.011$ & $0.360$ & $0.359$ & $0.428$ & $0.958$ \\
2 & $-0.097$ & $0.016$ & $0.515$ & $0.505$ & $0.861$ & $0.938$ \\
3 & $-0.019$ & $0.007$ & $0.234$ & $0.233$ & $0.240$ & $0.939$ \\
4 & $-0.045$ & $0.016$ & $0.503$ & $0.501$ & $1.087$ & $0.944$ \\
5 & $-0.017$ & $0.009$ & $0.286$ & $0.286$ & $0.284$ & $0.942$ \\
6 & $-0.014$ & $0.011$ & $0.359$ & $0.359$ & $0.386$ & $0.958$ \\
7 & $\phantom{-}0.004$ & $0.008$ & $0.243$ & $0.243$ & $0.261$ & $0.969$ \\
8 & $-0.004$ & $0.006$ & $0.180$ & $0.180$ & $0.181$ & $0.958$ \\
9 & $-0.025$ & $0.012$ & $0.375$ & $0.374$ & $0.415$ & $0.956$ \\
10 & $-0.016$ & $0.008$ & $0.242$ & $0.242$ & $0.257$ & $0.960$ \\
11 & $-0.024$ & $0.010$ & $0.328$ & $0.328$ & $0.353$ & $0.956$ \\
12 & $-0.013$ & $0.005$ & $0.167$ & $0.167$ & $0.165$ & $0.943$ \\
13 & $-0.022$ & $0.011$ & $0.341$ & $0.341$ & $0.371$ & $0.944$ \\
14 & $-0.004$ & $0.006$ & $0.186$ & $0.186$ & $0.194$ & $0.954$ \\
15 & $-0.010$ & $0.008$ & $0.244$ & $0.244$ & $0.252$ & $0.952$ \\
16 & $-0.013$ & $0.005$ & $0.172$ & $0.172$ & $0.179$ & $0.951$ \\
17 & $\phantom{-}0.003$ & $0.004$ & $0.122$ & $0.122$ & $0.125$ & $0.952$ \\
18 & $-0.019$ & $0.008$ & $0.255$ & $0.255$ & $0.265$ & $0.948$ \\
19 & $-0.002$ & $0.007$ & $0.223$ & $0.223$ & $0.221$ & $0.939$ \\
20 & $-0.019$ & $0.010$ & $0.305$ & $0.304$ & $0.304$ & $0.938$ \\
21 & $-0.007$ & $0.006$ & $0.194$ & $0.194$ & $0.197$ & $0.949$ \\
22 & $-0.023$ & $0.009$ & $0.277$ & $0.277$ & $0.287$ & $0.950$ \\
23 & $-0.003$ & $0.005$ & $0.147$ & $0.147$ & $0.152$ & $0.960$ \\
24 & $-0.006$ & $0.006$ & $0.187$ & $0.187$ & $0.198$ & $0.963$ \\
25 & $\phantom{-}0.010$ & $0.004$ & $0.142$ & $0.142$ & $0.143$ & $0.956$ \\
26 & $-0.003$ & $0.004$ & $0.131$ & $0.131$ & $0.136$ & $0.956$ \\
27 & $\phantom{-}0.010$ & $0.006$ & $0.205$ & $0.205$ & $0.207$ & $0.955$ \\
28 & $\phantom{-}0.003$ & $0.005$ & $0.147$ & $0.147$ & $0.152$ & $0.946$ \\
29 & $-0.014$ & $0.006$ & $0.196$ & $0.196$ & $0.203$ & $0.952$ \\
30 & $-0.003$ & $0.004$ & $0.138$ & $0.138$ & $0.136$ & $0.943$ \\
31 & $-0.003$ & $0.006$ & $0.193$ & $0.192$ & $0.194$ & $0.940$ \\
32 & $-0.005$ & $0.003$ & $0.104$ & $0.104$ & $0.106$ & $0.962$ \\
33 & $-0.001$ & $0.004$ & $0.129$ & $0.129$ & $0.134$ & $0.963$ \\
34 & $\phantom{-}0.010$ & $0.003$ & $0.099$ & $0.099$ & $0.099$ & $0.947$ \\
35 & $\phantom{-}0.001$ & $0.003$ & $0.096$ & $0.096$ & $0.095$ & $0.943$ \\
36 & $\phantom{-}0.001$ & $0.005$ & $0.148$ & $0.148$ & $0.144$ & $0.949$\\
\midrule
\end{tabular*}
\end{table}

\label{lastpage}